\newif\ifels
\elsfalse

\documentclass[runningheads]{llncs}

\usepackage{thm-restate}
\usepackage{mathpartir}
\usepackage{eurosym}
\usepackage{amsmath,amssymb}
\usepackage{xcolor}
\usepackage{stmaryrd,textcomp}
\usepackage{mathbbol}
\usepackage{mathtools}
\usepackage{graphicx}
\usepackage{upgreek}
\usepackage{tikz}
\usepackage{tikzscale}
\usepackage{url}
\usepackage{ifsym}
\usepackage{fancyvrb}
\usepackage{float}
\usepackage{listings}
\usepackage{color,soul}

\usepackage[mathscr]{euscript}

\newif\iflipics\lipicstrue

\newcommand{\semi}{\,,\,}

\renewcommand{\S}{\texttt{S}}

\newcommand{\bigfract}[2]{\frac{^{\textstyle #1}}{_{\textstyle #2}}}
\newcommand{\rulename}[1]{{\small {\sc[#1]}}}

\newcommand{\rulenamex}[1]{\mbox{\tiny [{\sc #1}]}}
\def \mathrule #1#2#3{	\begin{array}{l} 
                       	\rulenamex{#1}
                       	\\ 
                      	 \bigfract{#2}{#3}	
                       	\end{array}
					 }

\def \mathax #1#2{\begin{array}{l} 
                  \rulenamex{#1} 
                  \\ 
                  #2
                  \end{array}
                  }

\newcommand{\C}{{\tt C}}

\newcommand{\Q}{{\tt Q}}
\newcommand{\Qwithat}{\mbox{{\tt @}}\Q}
\newcommand{\withat}{\mbox{{\tt @}}}

\newcommand{\lolli}{\multimap}

\newcommand{\event}{\,\mbox{{\tt >\hspace{-2.1pt}>}}\,}

\newcommand{\tostate}{\,\mbox{{\tt =\hspace{-2.5pt}>}}\,}
\newcommand{\A}{{\tt A}}

\newcommand{\SemEvent}{\Psi}
\renewcommand{\time}{{\tt k}}

\newcommand{\now}{{\tt now}}

\newcommand{\timexpr}{\mbox{\normalsize {\tt t}}} 

\newcommand{\linecode}[2]{{\tt LC}_{#1}(#2)}
\newcommand{\ev}{{\tt ev}}
\newcommand{\miniStipula}{\mbox{{\sf \emph{$\mu$Stipula}}}}
\newcommand{\miniStipulaTA}{\mbox{{\sf \emph{$\mu$Stipula}}$^{\tt TA}$}}
\newcommand{\miniStipulaI}{\mbox{{\sf \emph{$\mu$Stipula}}$^{\tt I}$}}
\newcommand{\miniStipulaD}{\mbox{{\sf \emph{$\mu$Stipula}}$^{\tt D}$}}
\newcommand{\miniStipulaDI}{\mbox{{\sf \emph{$\mu$Stipula}}$^{\tt DI}$}}
\newcommand{\miniStipulaDIP}{\mbox{{\sf \emph{$\mu$Stipula}}$^{\tt DI}_{\mbox{\tt +}}$}}

\newcommand{\InitEv}{\mathtt{InitEv}}
\newcommand{\Stipula}{\mbox{{\sf \emph{Stipula}}}}

\newcommand{\X}{\mathtt{X}}

\newcommand{\B}{{\tt B}}

\newcommand{\contract}{\mathbb{C}}
\newcommand{\contractin}{\mathbb{C}_{\tt init}}
\newcommand{\TS}{\mathbb{TS}}
\newcommand{\TStp}{\mathbb{TS}_{\tt tp}}

\newcommand{\clause}[3]{#1 {\cdot} #2 {\cdot} #3}

\newcommand{\zero}{\mbox{\raisebox{-.6ex}{{\tt -\!\!\!-}}}}

\newcommand{\sem}[1]{\llbracket#1\rrbracket}
\newcommand{\curvedsembis}[1]{}

\newcommand{\R}{\mathtt{R}}
\newcommand{\f}{{\tt f}}
\newcommand{\g}{{\tt g}}

\newcommand{\Inc}{\textit{Inc}}
\newcommand{\DecJump}{\textit{DecJump}}

\newcommand{\lred}[1]{\mathrel{\stackrel{#1}{\longrightarrow}}}

\newcommand{\lredTickP}[1]{\mathrel{\stackrel{#1}{\longrightarrow}_{\tt tp}}}
\newcommand{\nolredTickP}[1]{\mathrel{\stackrel{#1}{\nrightarrow}_{\tt tp}}}

\newcommand{\noredbis}[1]{#1 \mathrel{\nrightarrow}}
\newcommand{\xred}[1]
{ \setbox0=\hbox{$\, {}^{#1}\, $}
  \setbox1=\hbox{$\longrightarrow$}
  \loop\setbox1=\hbox{$-$\kern-0.3em\unhbox1}\ifdim\wd1<\wd0\repeat
  \hbox{$\ \mathop{\box1}\limits^{#1} \ $}
}

\newcommand{\Pred}{\mathit{Pred}}

\newcommand{\n}{{\tt n}}

\newcommand{\eqdef}{\stackrel{\textsf{\tiny def}}{=}}

\newcommand{\pairbis}[2]{#1}

\newcommand{\Reviewer}[2]{}

\newcommand{\Answer}[1]{}

\ifels
\newtheorem{definition}{Definition}
\newtheorem{corollary}{Corollary}
\newtheorem{theorem}{Theorem}

\newtheorem{remark}{Remark}
\newproof{proof}{Proof}
\fi

\title{
Decidability Problems for Micro-Stipula
}

\author{G. Delzanno\inst{1} \and C. Laneve\inst{2} \and A. Sangnier\inst{1} \and  G. Zavattaro\inst{2}}

\institute{
DIBRIS, University of Genova, Italy
\and
DISI, University of Bologna, Italy
}

\authorrunning{Delzanno et al.}

\begin{document}

\maketitle

\begin{abstract} 
{\sf \emph{Micro-}}{\Stipula} is a stateful calculus in which clauses can be 
activated either through interactions with the external environment or by the evaluation 
of time expressions. Despite the apparent simplicity of its syntax and operational model, the combination of state evolution, time reasoning, and nondeterminism gives rise to significant analytical challenges. In particular, we show that determining whether a 
clause is never executed is undecidable. We formally prove that this undecidability 
result holds even for syntactically restricted fragments: namely, the time-ahead 
fragment, where all time expressions are strictly positive, and the instantaneous 
fragment, where all time expressions evaluate to zero. On the other hand, we identify 
a decidable subfragment: within the instantaneous fragment, reachability becomes 
decidable when the initial states of functions and events are disjoint.
%
%
\end{abstract}
  
\section{Introduction}
\label{sec:introduction}

{\sf \emph{Micro-}}{\Stipula}, noted {\miniStipula}, is a basic calculus defining \emph{contracts},
namely sets of clauses that are either (\emph{a}) \emph{parameterless functions}, to be
invoked by the external environment, or (\emph{b}) \emph{events} that are triggered at 
given times. The calculus has been devised to study the presence of clauses
in legal contracts written in {\Stipula}~\cite{CrafaL20,Stipula} that can never be applied 
because of unreachable circumstances or of wrong time constraints -- so-called \emph{unreachable clauses}.
In the legal contract domain, removing such clauses 
when the contract is drawn up is
substantial because they might be considered too oppressive by parties  and
make the legal relationship fail.

While dropping unreachable code is a very common optimization 
in compiler construction of programming languages~\cite{AhoUllman,Appel},
the presence of time expressions in {\miniStipula} events makes the optimization complex.
In particular, 
when the time expressions are logically inconsistent
with the contract behaviour,  the corresponding event (and its continuation) becomes
unreachable. 

In~\cite{Laneve04} we have defined an analyzer that uses symbolic expressions to approximate 
time expressions at static-time and that computes the set of reachable clauses by means of a 
closure operation based on a fixpoint technique. The analyzer, whose 
prototype is at~\cite{stipulaanalyzer}, is sound (every clause it spots is unreachable) but not 
complete (there may be unreachable clauses that are not recognized). 
The definition of a complete algorithm for unreachable clauses was left as an open problem.

In this paper  we study the foregoing problem for three fragments of {\miniStipula} that are 
used to model distinctive elements of legal contracts:

\Reviewer{1}{four fragments of MicroStipula are considered, which itself is a fragment of Stipula. It is unclear how these four fragments have been chosen. They are unmotivated and it is unclear how they relate to actual legal cases in practice. }
\Answer{We have motivated the fragments. The appendix (that is not part of the conference
paper) also reports more detailed motivations.}


\begin{description}
\item[$\miniStipulaTA: $] time expressions are always positive -- the corresponding events allow one to
define \emph{future obligations} such as
the power to exercise an option may only last until a deadline is met.
For example, a standard clause in a legal contract for renting a device is
\begin{itemize}
\item[--] \emph{The Borrower shall return the device {\tt k}
hours after the
rental and will pay Euro {\tt cost} in advance where half of the amount is of surcharge for late return. }
\end{itemize}
This clause is transposed in {\miniStipula} with an event like
\begin{center}
{\small {\tt \begin{tabular}{l}
      now + k $\event$ @Using \{
\\
\qquad  cost $\lolli$ Lender 
\\
\} $\tostate$ @End
\end{tabular}}}
\end{center}
which is affirming that the money {\tt cost} is sent to 
 the Lender (the operation $\lolli$) if the Borrower is {\tt Using} the device
when the deadline expires. 

\item[$\miniStipulaI: $] 
time expressions are always {\now} -- the corresponding events permit to define \emph{judicial enforcements} such as the \emph{immediate} activation of a dispute resolution mechanism
by an authority when a party challenges the content or the execution of the contract.
For example, a contract for renting a device might also contain a clause for 
resolution of problems:
\begin{itemize}
\item[--]
\emph{If the Lender detects a problem, he may trigger a resolution process that is 
managed by 
the Authority. The Authority will immediately enforce resolution to either the Lender or the Borrower.
}
\end{itemize}
In this case, the {\miniStipula} clause is the event
\begin{center}
{\small {\tt \begin{tabular}{l}
now $\event$ @Problem \{  
\\
\qquad \textcolor{olive}{ // immediately enforce resolution to Lender or
Borrower}
\\
\} $\tostate$ @Solved
\end{tabular}}}
\end{center}
where the \emph{immediate resolution} is implemented by a time expression {\tt now}. 
In this case, if the Lender and the Borrower have not yet resolved the issue -- the contract is in a 
state {\tt Problem} -- then the Authority enforces a resolution by, for example,
sending a part of {\tt cost} to Lender and the remaining part to Borrower. 

\item[$\miniStipulaD: $] functions and events do not have initial states
in common; events in this fragment allow one to
model \emph{exceptional behaviours} that must be performed \emph{before} any party may invoke a function. 
This type of restriction is quite common in legal contracts, 
particularly when formalizing a clause that outlines the consequences of a party breaching a condition.
Section~\ref{sec:Section2} contains a legal contract written in {\miniStipulaD}. 
\end{description}

We demonstrate that, even for the aforesaid fragments of {\miniStipula},
there is no complete algorithm for determining unreachable clauses. 
The proof technique is based on reducing the 
unreachability problem to the halting problem for Minsky machines (finite state machines
with two registers)~\cite{Minsky}. The 
{\miniStipula} encodings 
of Minsky machines model registers' values with multiplicities of events and extensively use the 
features that (\emph{a}) events preempt both the invocation of
functions and the progression of time and (\emph{b}) events that are not executed in the
current time slot are garbage-collected when the time progresses.
The encoding of {\miniStipulaTA} is particularly complex 
because we need to decouple in two different time slots the events corresponding to the two registers and
recreate them when the time progresses. 

We then restrict to {\miniStipulaDI}, a fragment of {\miniStipula} that is the intersection of {\miniStipulaI}
and {\miniStipulaD}, and demonstrate that the corresponding contracts are an instance of 
well-structured transition systems~\cite{Finkel:2001},
whereby the reachability problem is decidable. 
To achieve this result, we had to modify the semantics of {\miniStipula} by restricting the application 
of the time progression to states in which functions can be invoked (thus it is
disabled in the states where events are executed). 
Hereafter, the correspondence between the models using the two different progression rules has been 
analyzed to demonstrate reachability results for {\miniStipulaDI}.
\Reviewer{3}{
For the fragment of Micro-Stipula that enjoys decidability, additional motivations would be particularly useful: the more the fragment is capable of expressing a large class of legal contracts, the more useful the decidability result is. Regarding the last question, it would be interesting to know whether decidability can be lifted to the corresponding fragment of Stipula: if not, can the authors still argue for the practical relevance of their result? Finally, it would be useful to assess the real-world relevance of the problem of unreachable clauses within legal contracts.}
\Answer{
The breadth of legal contracts is indeed impressive and we can found examples motivating 
{\miniStipulaDI}. The second contract in the Appendix actually conforms to the 
constraints of this fragment. However, we find hard to motivate every fragment with 
the page limit we have. So, we think that it is reasonable to include longer motivations 
in the full paper.
}

The rest of this paper is structured as follows. Section~\ref{sec:Stipula} presents the calculus {\miniStipula}
with examples and the semantics. Section~\ref{sec:undecidability} contains the undecidability results for
{\miniStipulaTA}, {\miniStipulaI} and {\miniStipulaD}. Section~\ref{sec:det_instantaneous} contains
the decidability results
about {\miniStipulaDI}. Section~\ref{sec:related} reports 
and discusses related work and 
Section~\ref{sec:conclusions} presents general 
conclusions and future work. The proofs of our propositions, lemmas
and theorems are reported in Section~\ref{sec:proofs}.

\section{The calculus {\miniStipula}}
\label{sec:Stipula}

{\miniStipula} is a calculus of contracts. A contract is declared  by the term
\begin{center}
{\tt stipula C \{ 
        init $\Q$  \quad
        $F$ 
    \}
} 
\end{center}
where $\C$ is the name of the contract, $\Q$ is the \emph{initial state} and
a $F$ is a sequence of \emph{functions}. We use 
a set of \emph{states}, ranged over $\Q$, $\Q'$, $\cdots$; and a set of \emph{function
names} $\f$, $\g$, $\cdots$.
The above contract is defined by the keyword {\tt stipula} and is initially in the state specified by the 
{\tt init} keyword. 
The syntax of functions $F$, events $W$ and time expressions $\timexpr$ is 
the following:
\[
\begin{array}{l@{\quad}cc@{\quad}ll}
\mathit{Functions} & F & ::= & \zero  \quad | \quad   \Qwithat \;\f
\, \{
\,  
W\,\} \,\tostate\, \Qwithat'\; \; F 
\\
\mathit{Events} &  W & ::= & \zero  \quad | \quad \timexpr \,\event\, \Qwithat
\,\tostate\, \Qwithat'\; \;  W 
\\
\mathit{Time \; expressions} & \timexpr & ::= &  \now + {\tt k} & ({\tt k} \in {\sf Nat})
\end{array}
\]

Contracts transit from one state to another either 
by invoking a \emph{function} or by running an \emph{event}. 
Functions $\Qwithat \;\f \, \{\,  W\,\} \,\tostate\, \Qwithat'$ are invoked by the 
\emph{external environment}
and define the state $\Q$ when the invocation is admitted  and the state $\Q'$ when the 
execution of $\f$ terminates.

\emph{Events} $W$ are sequences of \emph{timed continuations} that are created by
functions and schedule a transition in future execution. More precisely, the term 
$\timexpr \event \Qwithat \tostate\,\Qwithat'$ schedules a transition from 
$\Q$ to  $\Q'$ at $\timexpr$ time slot ahead the current time if  the contract will be
in the state $\Q$. 
The time expressions 
are additions $\now + {\tt k}$, where 
{\tt k} is a constant (a natural number representing, for example,  \emph{minutes}); 
$\now$ is a place-holder
that  will be replaced by 0 during the execution, see rule
\rulename{Function} in Table~\ref{tab.StipulaSemantics}. We always shorten $\now + 
{\tt 0}$
into $\now$.

\paragraph{Restriction and notations.} 
We write $\Qwithat \;\f\, \{ \, W\,\} \,\tostate\, \Qwithat' \in \C$ when the function
$\Qwithat \;\f\, \{ \, W\,\} \,\tostate\, \Qwithat'$ is in the contract $\C$. Similarly for events.
We assume that \emph{a function is uniquely determined by the 
tuple}
$\clause{\Q}{\f}{\Q'}$, that is the initial and final states
and the function name. 
In the same way, an event is uniquely determined by the tuple $\clause{\Q}{\ev_\n}{\Q'}$,
where {\tt n} is the line-code of the event~\footnote{We assume the code of
{\miniStipula} contracts to be organized in lines
of code, and each line contains at most one event definition.}.
Functions and events are generically called \emph{clauses} and, since tuples 
$\clause{\Q}{\f}{\Q'}$ and $\clause{\Q}{\ev_\n}{\Q'}$ uniquely identify functions and events, we 
will also call them clauses and write $\clause{\Q}{\f}{\Q'} \in \C$ and $\clause{\Q}{\ev_\n}{\Q'} \in \C$.
%

\subsection{Examples}
As a first, simple example consider the {\tt PingPong} contract:
 
{\small
\begin{lstlisting}[xleftmargin=.1\textwidth,linewidth=.7\textwidth, numbers=left,mathescape,basicstyle=\ttfamily]
stipula PingPong {	
   init Q0
   @Q0 ping {
       now + 1 $\event$ @Q1 $\tostate$ @Q2
   } $\tostate$ @Q1
   @Q2 pong { 
       now + 2 $\event$ @Q3 $\tostate$ @Q0
   } $\tostate$ @Q3
}
\end{lstlisting}
}

\noindent
The contract contains two functions: {\tt ping} and {\tt pong}. In particular {\tt ping} 
is invoked if the contract is in the state {\tt Q0}, {\tt pong} when the
contract is in {\tt Q2}. Functions (\emph{i}) make the contract transit in the 
state specified
by the term ``$\tostate \Qwithat$'' (see lines {\tt 5} and {\tt 8}) and (\emph{ii})
make the events in their body to be scheduled. In particular, an event
{\tt now + k $\event$ @Q $\tostate$ @Q$'$}
(see lines {\tt 4} and {\tt 7}) is a timed 
continuation that can run when the time is {\tt k} \emph{time slots ahead 
to the clock value when the function is called} and the state
of the contract is {\tt Q}. The only effect of executing an event is the change of
the state. When no event can be executed in a state either \emph{the time progresses} 
(a tick occurs) 
or \emph{a function is invoked}. The progression of time does not modify a state.

In the {\tt PingPong} contract, the initial state is {\tt Q0} where only
{\tt ping} may be invoked; no event is present because they are created by executing 
functions. The invocation of {\tt ping} makes the contract transit 
to {\tt Q1} and creates the event at line {\tt 4}, noted ${\tt ev}_{\tt 4}$. 
In {\tt Q1} there is still a unique possibility: executing ${\tt ev}_{\tt 4}$.
However, to execute it, it is necessary to wait {\tt 1} minute (one clock tick must elapse)
-- the time expression {\tt now + 1}. Then the state becomes {\tt Q2}
indicating that {\tt pong} may be invoked, thus letting the contract transit
to {\tt Q3} where, after {\tt 2} minutes (the expression {\tt now + 2}),
the event at line {\tt 7} can be executed and the contract returns to {\tt Q0}.
In {\tt PingPong}, every clause is reachable.

The following {\tt Sample} contract has an event that is unreachable:

{\small
\begin{lstlisting}[xleftmargin=.1\textwidth,linewidth=.7\textwidth, numbers=left,mathescape,basicstyle=\ttfamily]
stipula Sample {	
   init Init
   @Init f {
       now + 0 $\event$ @Go $\tostate$ @End
   } $\tostate$ @Run
   @Init g { } $\tostate$ @Go
}
\end{lstlisting}
}

\noindent
Let us discuss the issue. {\tt Sample} has two functions at lines {\tt 3} and {\tt 6}, called $\f$ and $\g$, respectively.
The two functions may be invoked in {\tt Init}, however
the
invocation of one of them excludes the other because their final states are
not {\tt Init}. Therefore 
the event at line {\tt 4}, which is inside {\tt f}, is unreachable since it can run 
only if $\g$ is executed. 

\subsection{The operational semantics}
\label{sec:operationalsemantics}
The meaning of {\miniStipula} primitives is defined operationally by a transition 
relation between configurations. A configuration, ranged over by $\contract$, $\contract'$, $\cdots$, is a tuple
$\C(\Q\semi \Sigma \semi\Psi)$ where 
\begin{itemize}
\item
$\C$ is the contract name;
\item 
$\Q$ is the current state of the contract;
\item
$\Sigma$ is either 
$\zero$ or a term  $\Psi \tostate \Q$. 
$\Sigma$ represents either an empty body (hence, a
clause can be executed or the time may progress) or a continuation where
a set of events $\Psi$ must be evaluated;
\item
$\Psi$ is a (possibly empty) multiset of \emph{pending events} that have been already scheduled 
for future execution but not yet triggered. 
In particular,
$\Psi$ is either $\zero$, when 
there are no pending events, or it is $\time_1\! \event_{\!{\tt n}_1\;}\! \Q_1 \tostate \Q_1' ~|\! \cdots \!|~ 
\time_h \!\event_{\!{\tt n}_h\;}\! \Q_h \tostate \Q_h'$ where ``$|$'' is commutative and associative
 with identity $\zero$.  
In every term $\time_i \event_{\!{\tt n}_i\;} \Q_i \tostate \Q_i'$, the constant $\time_i$ 
is obtained from the
time expression $\timexpr_i$ of the corresponding event by dropping  $\now$. The index $\n_i$ is
the \emph{line-code} of the event.

The function that turns a \emph{sequence of events} $\now + {\tt k} \event \Qwithat \tostate \Qwithat'$ into
a \emph{multiset of terms} ${\tt k} \event_{\!\n}\; \Q \tostate \Q'$ is
$\linecode{\clause{\Q }{\f}{\Q'}}{W}$ (see rule~\rulename{Function}, 
the trivial definition of this function is omitted). 
This function also drops the ``{\tt @}'' from the states.
\end{itemize}       
%
The transition relation of {\miniStipula} is  
$\pairbis{\contract}{\time} \lred{\mu} \pairbis{\contract'}{\time'}$, 
where $\mu$ is either \emph{empty} $\_$ or $\texttt{f}$ or $\ev_\n$ 
(the label $\ev_\n$ indicates the event at line $\n$).
The formal definition of $\pairbis{\contract}{\time} \lred{\mu} \pairbis{\contract'}{\time'}$ 
is given in Table~\ref{tab.StipulaSemantics} using 
\begin{itemize}
\item
the predicate $\noredbis{\Psi, \Q}$, whose definition is
\[
\noredbis{\Psi, \Q} \; \eqdef \;
\left\{
	\begin{array}{l@{\qquad}l}
	{\it true} & \text{if } \; \Psi = \zero
	\\
	{\it false} & \text{if } \; \Psi = 0 \event_{\!{\tt n}\;} \Q \tostate \Q'~|~ \Psi'
	\\
	\noredbis{\Psi', \Q} & \text{if } \; \Psi = {\tt k} \event_{\!{\tt n}\;} \Q' \tostate \Q'' ~|~ \Psi'
	\; \text{ and } \; ({\tt k} \neq 0 \; \text{ or } \; \Q' \neq \Q)
	\end{array} \right.
\]
\item
the function $\Psi \downarrow$, whose definition is
\[
\begin{array}{rll}
(\Psi ~|~ \Psi') \downarrow \; = & \Psi\downarrow ~|~ \; \Psi'\!\downarrow
\\
(\time+1 \event_{\!{\tt n}\;} \Q' \tostate \Q'') \downarrow \; = &  
\time \event_{\!{\tt n}\;} \Q' \tostate \Q''
\\
(0 \event_{\!{\tt n}\;} \Q' \tostate \Q'') \downarrow \; = &  \zero 
\end{array}
\] 
\end{itemize}

A discussion about the four rules follows. Rule \rulename{Function} defines invocations: 
the label specifies the function name {\tt f}. The 
transition may occur provided (\emph{i}) the contract is in the state \texttt{Q} that admits invocations of
{\tt f} and (\emph{ii}) no event can be triggered -- \emph{cf.}~the premise $\noredbis{\SemEvent, \Q}$ (event's execution preempts function invocation). 
Rule \rulename{State-Change} says that a contract changes state by adding
 the sequence of events 
$W$ to the multiset of pending events once {\now} has been dropped from time expressions.
Rule \rulename{Event-Match} specifies that an event handler may run provided $\Sigma$ is 
$\zero$, the time guard of the event has value $0$ and the initial state of the event is the same of 
the contract's state.
%
Rule \rulename{Tick} defines the progression of time. This happens when the contract 
has an empty $\Sigma$ 
and no event can be triggered. In this case, the events with time value 0 are garbage-collected and the 
the time values of the other events are decreased by one.
\begin{table}[t]
\[
\begin{array}{c@{\qquad}c}
\mathrule{Function}{
	\begin{array}{c}
	\Qwithat\,\f  
	\texttt{\{}\,W\, \texttt{\}} \tostate  \Qwithat' \in \C
	\quad
	\Psi' = \linecode{\clause{\Q}{\f}{\Q'}}{W}
	\\
	\noredbis{\SemEvent, \Q}
	\end{array}
	}{
	\pairbis{\C(\Q \semi  \zero \semi \SemEvent)}{\time} 
	    \lred{\f}
	\pairbis{\C(\Q \semi \Psi' \tostate  \Q' 
	\semi \SemEvent)}{\time}
	}
&
\mathrule{Event-Match}{
	\begin{array}{c}
	\SemEvent = {\tt 0} \event_{\tt \! n \;}  \Q \tostate  \Q' ~|~ \SemEvent'
	\end{array}
	}{
	\pairbis{\C(\Q \semi \zero \semi \SemEvent)}{\time} \lred{\ev_\n}
	\pairbis{\C(\Q \semi \zero \tostate  \Q' \semi \SemEvent')}{\time}
	} 
\\[.9cm]
\mathax{State-Change}{
	\pairbis{\C(\Q \semi \Psi' \tostate  \Q' \semi \SemEvent)}{\time} 
	\lred{}
	\pairbis{\C(\Q' \semi \zero \semi \SemEvent' ~|~\SemEvent)}{\time}
	}
&
\mathrule{Tick}{
	\noredbis{\SemEvent, \Q}
	}{
	\pairbis{\C(\Q \semi \zero \semi \SemEvent)}{\time} \lred{} 
	\pairbis{\C(\Q  \semi \zero \semi \SemEvent \downarrow)}{\time + 1}
	}
\\
\\
\end{array}
\]
\caption{\label{tab.StipulaSemantics} The operational semantics of {\miniStipula}}
\vspace{-.6cm}
\end{table}
The rules~\rulename{Function} and~\rulename{State-Change} might have been 
squeezed in one
rule only. We have preferred to keep them apart for compatibility with {\Stipula} (where 
functions' and events' bodies may also contain statements). 

The \emph{initial configuration} of a {\miniStipula} contract 
\begin{center}
{\tt stipula C \{ 
        init $\Q$  \quad
        $F$ 
    \}
} 
\end{center}
is $\contractin =  \pairbis{\C(\Q \semi \zero \semi \zero)}{\time}$; 
the \emph{set of configurations} of $\C$ are denoted by $\mathcal{C}_\C$. 

We write $\pairbis{\contract}{\time} \lred{} \pairbis{\contract'}{\time'}$ if there is a $\mu$
(as said above, $\mu$  may be either $\zero$ or a function 
name  or an event) such that $\pairbis{\contract}{\time} \lred{\mu} \pairbis{\contract'}{\time'}$.
We also write
$ \pairbis{\contract}{\time}\lred{}^* \pairbis{\contract'}{\time'}$, called \emph{computation},
if there are $\mu_1, \cdots, \mu_h$ such that $\pairbis{\contract}{\time} \lred{\mu_1} \cdots \lred{\mu_h}
\pairbis{\contract'}{\time'}$. 
\Reviewer{1}{Is there a reason to mix labelled and unlabelled transitions?}
\Answer{done}
Labels are useful in the examples (and proofs) to highlight the 
clauses that are executed. However, while in~\cite{Laneve04},  they were introduced to 
ease the formal reasonings, labels be overlooked in this work.
Let $\TS(\C)=(\mathcal{C}_\C,\lred{})$ be the transition system associated to $\C$.

\begin{remark}
Few issues about the semantics in Table~\ref{tab.StipulaSemantics} are worth to
remarked.
\begin{itemize}
\item[--]
{\miniStipula} has three \emph{causes for nondeterminism}: 
(\emph{i})
two functions can be invoked in a state, (\emph{ii}) either a function is invoked or
the time progresses (in this case the function may be invoked at a later time), and
(\emph{iii}) if two events may be executed at the same time, then one is chosen and
executed.

\item[--]
Rule~\rulename{Tick} defines the progression of time. This happens when the contract has
no event to trigger. 
Henceforth, the complete execution of a function or of an event cannot last more than a single time unit. It is worth to notice that this semantics admits the paradoxical phenomenon 
that an endless sequence of function invocations does not make time progress 
(\emph{cf.}~the encoding of the 
Minsky machines in {\miniStipulaI}).
\Reviewer{2}{"The phenomenon... nondeterministic." is not
   clear. What do you mean with in "follows by rules that are more
   nondeterministic"? 
}
\Answer{The sentence has been replaced.}
This paradoxical behaviour, which is also present in process calculi with 
time~\cite{HanssonJ90,MollerT90}, might be removed by adjusting
the semantics so to progress time when a maximal number of functions has been
invoked. To ease the formal arguments we have preferred to stick to the simpler semantics.
\item[--]
The semantics of {\miniStipula} in this paper is different from, yet equivalent to,~\cite{Laneve04,Stipula}.
In the literature, configurations have clock values that are incremented by the tick-rule. Then,
in the \rulename{Function} rule, the variable $\now$ is replaced by the current clock value 
(and not dropped, as in our rule). We have chosen the current presentation because it eases the reasoning
about expressivity.  
\end{itemize}
\end{remark}
%

To illustrate {\miniStipula} semantics, we discuss the computations of the 
{\tt PingPong} contract. Let 
${\tt Ev}_4 = 1 \event_{\!4\;} \Q{\tt 1} \tostate \Q{\tt 2}$ and
${\tt Ev}_7 = 2 \event_{\!7\;} \Q{\tt 3} \tostate \Q{\tt 0}$. We  write ${{\tt Ev}_i}^{(-j)}$ to indicate the 
${\tt Ev}_i$ where the time guard has been decreased by $j$ (time) units.

The contract 
may initially perform a number of \rulename{Tick} transitions, say $k$, 
and then a \rulename{Function} one. Therefore we have
(on the right we write the rule that is used):
\[
\begin{array}{r@{\;\;}l@{\qquad\qquad}l}
\pairbis{{\tt PingPong}({\tt Q0} \semi \zero \semi \zero)}{\time} \; \lred{}^k & \pairbis{{\tt PingPong}({\tt Q0} \semi \zero \semi \zero)}{\time'}
& \mbox{\rulename{Tick}}
\\
\lred{{\tt ping}} & \pairbis{{\tt PingPong}({\tt Q0} \semi {\tt Ev}_4 \tostate {\tt Q1} \semi \zero)}{\time'}
& \mbox{\rulename{Function}}
\\
\lred{} & \pairbis{{\tt PingPong}({\tt Q1} \semi \zero \semi {\tt Ev}_4)}{\time'}
& \mbox{\rulename{State-Change}}
\\
\lred{} & \pairbis{{\tt PingPong}({\tt Q1} \semi \zero \semi {{\tt Ev}_4}^{(-1)} )}{\time'+1}
& \mbox{\rulename{Tick}}
\\
\lred{\ev_4} & \pairbis{{\tt PingPong}({\tt Q1} \semi \zero \tostate {\tt Q2}\semi \zero )}{\time'+1}
& \mbox{\rulename{Event-Match}}
\\
\lred{} & \pairbis{{\tt PingPong}({\tt Q2} \semi \zero \semi \zero)}{\time'+1}
&  \mbox{\rulename{State-Change}}
\end{array}
\]
In ${\tt PingPong}({\tt Q2} \semi \zero \semi \zero)$, the 
contract 
may perform a \rulename{Function} executing {\tt pong}. Then
the computation continues as follows:
\[
\begin{array}{@{\qquad \qquad \qquad \qquad}c@{\; \;}l@{\qquad \qquad}l}
\lred{{\tt pong}} & \pairbis{{\tt PingPong}({\tt Q2} \semi {\tt Ev}_7 \tostate {\tt Q3} \semi \zero)}{\time''}
&  \mbox{\rulename{Function}}
\\
\lred{} & \pairbis{{\tt PingPong}({\tt Q3} \semi \zero \semi {\tt Ev}_7 )}{\time''}
& \mbox{\rulename{State-Change}}
\\
\lred{}^2 & \pairbis{{\tt PingPong}({\tt Q3} \semi \zero \semi {{\tt Ev}_7 }^{(-2)} )}{\time''+2}
& \mbox{\rulename{Tick}}
\\
\lred{\ev_7} & \pairbis{{\tt PingPong}({\tt Q3} \semi \zero \tostate {\tt Q0}\semi \zero)}{\time''+2}
& \mbox{\rulename{Event-Match}}
\\
\lred{} & \pairbis{{\tt PingPong}({\tt Q0} \semi \zero \semi \zero )}{\time''+2}
& \mbox{\rulename{State-Change}}
\end{array}
\]

\begin{definition}[State reachability]
\label{def:reachability}
Let $\C$ be a {\miniStipula} contract with initial configuration $\contractin$.
A state $\Q$ is reachable in $\C$ if and only if there exists a configuration 
$\pairbis{\C(\Q \semi \zero \semi \Psi)}{\time'}$ such that $\contractin \lred{}^* \pairbis{\C(\Q \semi \zero \semi \Psi)}{\time'}$.
\end{definition}


It is worth noting that our notion of state reachability is similar to 
the notion of control state reachability 
introduced by Alur and Dill 
in the context of timed automata in \cite{AD94}.
Control state reachability is defined as the problem
of checking, given an automaton
$A$ and a control state $q$, if there 
exists a run of $A$ that visits $q$.
Control state reachability was studied also 
for lossy Minsky Machines in \cite{Mayr03}
and lossy FIFO channel systems in \cite{PJ93,CFI96}.
In the context of Petri Nets it can be reformulated in terms of 
coverability of a given target marking \cite{KM69}.

%

\subsection{Relevant sublanguages}
We will consider the following fragments of {\miniStipula} whose relevance has been already discussed in 
the Introduction:

\Reviewer{1}{this should be more formal.}\Answer{Done}

\begin{description}
\item[{\miniStipulaI}{\rm ,}] 
called instantaneous {\miniStipula}, is the fragment where every time expression 
of the events is {\now} {\tt + 0}; 

\item[{\miniStipulaTA}{\rm ,}] 
called time-ahead {\miniStipula}, is the fragment where every time expression of the events 
is {\tt now + k}, with ${\tt k} > {\tt 0}$;

\item[{\miniStipulaD}{\rm ,}]
called determinate {\miniStipula}, is the fragment where 
the sets of initial states of functions and of events have empty intersection;
\emph{i.e.}, for each function
$\Qwithat \;\f\, \{\,  W\,\} \,\tostate\, \Qwithat'$
and event
$\timexpr \,\event\, \Qwithat''\,\tostate\, \Qwithat'''$
in a contract, we impose $\Q \neq \Q''$;

\item[{\miniStipulaDI}{\rm ,}]
called determinate-instantaneous {\miniStipula}, is the 
intersection between {\miniStipulaD} and {\miniStipulaI};
\emph{i.e.}, for each function
$\Qwithat \;\f\, \{\,  W\,\} \,\tostate\, \Qwithat'$
and event
$\timexpr \,\event\, \Qwithat''\,\tostate\, \Qwithat'''$
in a contract, we impose $\Q \neq \Q''$
and $\timexpr = {\tt now}+{\tt 0}$.
\end{description}

\section{Undecidability results}
\label{sec:undecidability}

To show the undecidability of state reachability, we rely on  a reduction technique from a Turing-complete model to {\miniStipulaI},
{\miniStipulaTA},  
and  {\miniStipulaD}. 
The Turing-complete models we consider are the Minsky machines~\cite{Minsky}. 
A Minsky machine is an automaton with two registers $\R_1$ and $\R_2$ holding arbitrary large 
natural numbers, a finite set of states $\Q, \Q', \cdots$, and a program {\tt P} consisting of a 
finite sequence of numbered instructions 
of the following type:
\begin{itemize}
\item
$\Q : \Inc(\R_i, \Q')$: in the state $\Q$, increment $\R_i$ 
and go to the state $\Q'$;
\item
$\Q : \DecJump(\R_i, \Q', \Q'')$: in the state $\Q$, if the content of $\R_i$ is zero then
go to the state $\Q'$, else decrease $\R_i$ by 1 and go to the
state $\Q''$.
\end{itemize}
A configuration of a Minsky machine is given by a tuple $(\Q,v_1,v_2)$ where $\Q$ indicates the 
state of the machine and $v_1$ and $v_2$ are the contents of the two registers. A transition of a
Minsky machine is denoted by $\lred{}_{\tt M}$.
We assume that the 
machine has an \emph{initial state} $\Q_0$ and a \emph{final state} $\Q_F$ that has no instruction
starting at it.  The \emph{halting problem} of a Minsky machine is assessing whether there exist $v_1, v_2$ 
such that $(\Q_F, v_1, v_2)$ is reachable 
starting from $(\Q_0, 0, 0)$. This problem is undecidable~\cite{Minsky}. In the rest of the section, we will demonstrate the undecidability of state reachability for 
{\miniStipulaI}, {\miniStipulaTA},  
and  {\miniStipulaD} by providing encodings of Minsky machines. The encodings we use are increasingly complex. Therefore, we will present the undecidability results starting from the simplest one.

\subsection{Undecidability results for {\miniStipulaI}}
\label{sec:instantaneous}

\begin{table}[t]
\hrulefill

Let $M$ be a Minsky machine with initial state $\Q_0$. Let {\tt I}$_M$ be the
{\miniStipulaI} contract
\begin{center}
{\tt stipula I$_M$ \{ 
        init Start  \quad
        $F_M$ 
    \}
} 
\end{center}
where $F_M$ contains the functions 
\begin{itemize}
\item
{\tt @Start fstart $\{$ \now $\event$ @aQ$_0$ $\tostate$ @bQ$_0$ $\}$ $\tostate \; \Qwithat_0$  } 

\item
for every instruction $\Q : \Inc(\R_i, \Q')$, with $i \in \{1, 2 \}$: 
\[
\begin{array}{l@{\;}l}
\Qwithat \; \; {\tt fincQ} \; \{ & \now \, \event \, \withat{\tt dec}_i \tostate \withat{\tt ackdec}_i
\\
& \now \, \event \, \withat{\tt bQ} \tostate \Qwithat'
\\
& \now \, \event \, \withat{\tt aQ}' \tostate \withat{\tt bQ}'
\\
 \} \tostate \withat{\tt aQ} 
\end{array}
\]

\medskip

\item
for every instruction $\Q : \DecJump(\R_i, \Q', \Q'')$,  with $i \in \{1, 2 \}$:
\[
\begin{array}{l@{\qquad\qquad}l}
\begin{array}{l@{\;}l}
\Qwithat \; \; {\tt fdecQ} \;  \{ & \now \, \event \, \withat{\tt ackdec}_i \tostate \withat{\tt aQ}
\\
& \now \, \event \, \withat{\tt bQ} \tostate \Qwithat''
\\
& \now \, \event \, \withat{\tt aQ}'' \tostate \withat{\tt bQ}''
\\
 \} \tostate \withat{\tt dec}_i
\end{array}
&
\begin{array}{l@{\;}l}
\Qwithat \; \; {\tt fzeroQ} \, \{ & \now \, \event \,  \withat{\tt zero}_i \tostate \withat{\tt aQ} 
\\
&\now \, \event \, \withat{\tt bQ} \tostate \Qwithat'
\\
& \now  \, \event \, \withat{\tt aQ}' \tostate \withat{\tt bQ}'
\\
 \}
\tostate \withat{\tt dec}_i
\end{array}
\end{array}
\]
\item
$\withat{\tt dec}_1 \; \;  {\tt fdec1} \; \{ \; \} \; \tostate \; \withat{\tt zero}_1$ \quad and \quad
$\withat{\tt dec}_2 \; \;  {\tt fdec2} \; \{ \; \} \; \tostate \; \withat{\tt zero}_2$
\end{itemize}
\hrulefill

\caption{\label{tab.encodingI} The {\miniStipulaI} contract modelling a Minsky machine}
\end{table}
Table~\ref{tab.encodingI} defines the encoding of a Minsky machine $M$ into a 
{\miniStipulaI} contract ${\tt I}_M$.
The relevant invariant of the encoding is that, every time $M$ transits to
$(\Q, v_1, v_2)$ then {\tt I}$_M$ may transit to
${\tt I}_M(\Q, \zero, \Psi)$, where the number of events $0 \, \event \, \withat\mathtt{dec}_1 \tostate \withat{\tt ackdec}_1$ and $0 \, \event \, \withat\mathtt{dec}_2 \tostate \withat{\tt ackdec}_2$ in $\Psi$
are $v_1$ and $v_2$, respectively. Additionally, a transition 
$(\Q, v_1, v_2) \lred{}_{\tt M} (\Q',v_1', v_2')$ corresponds to a sequence of transitions 
${\tt I}_M(\Q, \zero, \Psi) \lred{}^* {\tt I}_M(\Q', \zero, \Psi')$
with either (\emph{i}) a {\tt fincQ} function, if the Minsky machine performs an ${\it Inc}$ instruction, 
or (\emph{ii}) either a {\tt fdecQ} or a {\tt fzeroQ} function, if the Minsky machine performs a
 ${\it DecJump}$ instruction. 
In particular, the function {\tt fincQ} has the ability to add one instance
of the event $0 \, \event \, \withat\mathtt{dec}_1 \tostate \withat{\tt ackdec}_1$ (or $0 \, \event \, \withat\mathtt{dec}_2 \tostate \withat{\tt ackdec}_2$). The function {\tt fdecQ} has the effect of consuming one
instance of the event $0 \, \event \, \withat\mathtt{dec}_1 \tostate \withat{\tt ackdec}_1$ (resp. $0 \, \event \, \withat\mathtt{dec}_2 \tostate \withat{\tt ackdec}_2$), by entering the state $\withat\mathtt{dec}_1$
(resp. $\withat\mathtt{dec}_2$) which triggers such event. Also the function 
{\tt zeroQ} enters in one of the states $\withat\mathtt{dec}_i$, 
but in this case the computation will have the ability to continue
only if the state $\withat\mathtt{zero}_i$ will be reached
(this because the produced 
event $0 \, \event \, \withat\mathtt{zero}_i \tostate \withat{\tt aQ}$
must be triggered to allow the computation to continue). 
But $\withat\mathtt{zero}_i$ can be reached only
if no event $0 \, \event \, \withat\mathtt{dec}_i$ 
is present, because only in this case the function {\tt fdec$_i$}
can be invoked (remember that events have priority w.r.t. 
function invocations).

We finally observe that the transitions 
${\tt I}_M(\Q, \zero, \Psi) \lred{}^* {\tt I}_M(\Q', \zero, \Psi')$
which mimick the Minsky machine step $(\Q, v_1, v_2) \lred{}_{\tt M} (\Q',v_1', v_2')$ are not the unique possibile transitions.
Nevertheless, in case different alternative transitions
are executed, the contract {\tt I}$_M$ will have no longer
the possibility to reach a state corresponding to a Minsky 
machine state, thus the simulation of the machine gets stuck. 
One of the alternative transitions is the invocation of the
function {\tt fdecQ} when the corresponding register is empty.
In this case, the simulation gets stuck because the state
$\withat{\tt ackdec}_i$, necessary to trigger the event  
$0 \, \event \, \withat{\tt ackdec}_i \tostate \withat{\tt aQ}$,
cannot be reached.
Similarly, 
if {\tt fzeroQ} is invoked when the corresponding register is nonempty
the state $\withat{\tt zero}_i$, necessary to trigger the event  
$0 \, \event \, \withat{\tt zero}_i \tostate \withat{\tt aQ}$,
cannot be reached.
Also the elapsing of time
is problematic because the events 
$0 \, \event \, \withat\mathtt{dec}_1 \tostate \withat{\tt ackdec}_1$ 
and $0 \, \event \, \withat\mathtt{dec}_2 \tostate \withat{\tt ackdec}_2$, which model the 
content of the registers, are erased by a \rulename{Tick} transition.
This could corrupt the modeling of the registers.
In this case the simulation gets stuck because 
also the ``management event'' $0 \, \event \, {\tt aQ} \tostate {\tt bQ}$ 
is erased, which is necessary to model the transition from a state {\tt Q}
of the Minsky machine to the next one.

\begin{restatable}{theorem}{thmintantaneous}
\label{thm.instantaneous}
State reachability is undecidable in {\miniStipulaI}.
\end{restatable}

\subsection{Undecidability results for {\miniStipulaTA}}
\label{sec:timeahead}
\newcommand{\control}{management}

%
Also in this case we reduce 
from the halting problem of Minsky machines to state
reachability in {\miniStipulaTA}. 
The encoding of a machine $M$ is defined in Table~\ref{tab.encodingTA}.
\begin{table}[t]

\hrulefill

Let $M$ be a Minsky machine with initial state $\Q_0$. Let {\tt TA}$_M$ be the
{\miniStipulaTA} contract
\begin{center}
{\tt stipula TA$_M$ \{ 
        init Q$_0$  \quad
        $F_M$ 
    \}
} 
\end{center}
with $F_M$ containing the functions 
\begin{itemize}
%
\item
for every instruction $\Q : \Inc(\R_i, \Q')$, with $i \in \{1, 2 \}$: 
\[
\begin{array}{l@{\;}l}
\Qwithat \; \; {\tt fincQ} \; \{ & \now+1 \, \event \, \withat{\tt dec}_i \tostate \withat{\tt ackdec}_i
\\
& \now+1 \, \event \, \Qwithat' \tostate \withat{\tt end}
\\
\} \tostate \Qwithat' 
\end{array}
\]

\medskip
\item
for every instruction $\Q : \DecJump(\R_i, \Q', \Q'')$,  with $i \in \{1, 2 \}$:
\[
\begin{array}{l@{\quad}l}
\begin{array}{l@{\;}l}
\Qwithat \; \; {\tt fdecQ} \;  \{ & \now +1\, \event \, \withat{\tt ackdec}_i \tostate {\withat{\tt nextQ''}}_i
\\
& \now +1 \, \event \, \withat{\tt wait} \tostate \withat{\tt dec}_1
\\
& \now+2 \, \event \, \withat{\tt dec}_1 \tostate \withat{\tt end}
\\
& \now+2 \, \event \, \withat{\tt dec}_2 \tostate \withat{\tt end}
\\
& \now+2 \, \event \, {\withat{\tt nextQ''}}_i \tostate \withat{\tt end}
\\
& \now+2 \, \event \, \withat{\tt ackdec}_1 \tostate \withat{\tt end}
\\
& \now+2 \, \event \, \withat{\tt ackdec}_2 \tostate \withat{\tt end}
\\
& \now+3 \, \event \, {\withat{\tt Q''}} \tostate \withat{\tt end}
\\
 \} \tostate \withat{\tt wait}
\end{array}
&
\begin{array}{l@{\;}l}
\Qwithat \; \; {\tt fzeroQ} \, \{ & \now+1 \, \event \,  \withat{\tt ackdec}_i \tostate \withat{\tt end} 
\\
&\now+2 \, \event \, \withat{\tt next} \tostate {\withat{\tt Q'}}
\\
& \now +1 \, \event \, \withat{\tt wait} \tostate \withat{\tt dec}_1
\\
& \now +3 \, \event \, {\withat{\tt Q'}} \tostate \withat{\tt end}
\\
 \}
\tostate \withat{\tt wait}
\\
\\
\\
\\
\\
\end{array}
\end{array}
\]

\item
for $i \in \{1,2\}$ and for every state $\Q$ of $M$, we have the following functions:
\[
\begin{array}{l}
\withat{\tt wait} \; \;  {\tt fwait} \; \{ \; \} \; \tostate \; \withat{\tt end}
\\
\withat{\tt dec}_1 \; \;  {\tt fdec1} \; \{ \; \} \; \tostate \; \withat{\tt dec}_2
 \quad \text{ and } \quad
\withat{\tt dec}_2 \; \;  {\tt fdec2} \; \{ \; \} \; \tostate \; \withat{\tt next}
\\
\withat{\tt ackdec}_i \; \;  {\tt fackdec\_i} \; \{ \; \now + 2 \event \withat{\tt dec}_i \tostate
\withat{\tt ackdec}_i \} \; \tostate \; \withat{\tt dec}_i
\\
{\withat{\tt nextQ}}_i \; \;  {\tt fnextQ\_i} \; \{ \; \now + 1 \event \withat{\tt next} \tostate
\Qwithat \} \; \tostate \; \withat{\tt dec}_i
\end{array}
\]
\end{itemize}
\hrulefill
\caption{\label{tab.encodingTA} The {\miniStipulaTA} contract modelling a Minsky machine}
\end{table}
\Reviewer{1}{the distinction between machine and control state is unclear.}
\Answer{Thanks for the comment; the terminology ``control state'' is not
fortunate due to the overloading w.r.t. ``control state reachability''
mentioned in other parts of the paper. We have modified ``control state''
in ``{\control} state''}
In this case,
states of the {\miniStipulaTA} contract alternates between ``\emph{machine states}'' occurring, say, at even 
time clocks, and ``\emph{{\control} states}'' occurring at odd time clocks. For this reason we add
erroneous transitions at even time clocks from {\control} states to {\tt end} (a state without 
outgoing transitions) and at odd time clocks from machine states to {\tt end}.
Similarly to Table~\ref{tab.encodingI}, a unit in the register $i$ is encoded by an event
$\now + 1 \event \withat{\tt dec}_i \tostate \withat{\tt ackdec}_i$ (assuming to be in a machine state).
Therefore the instruction $\Q : \Inc(\R_i, \Q')$, which occurs in a machine state $\Q$, amounts to adding to the next time-clock
such event and 
the erroneous event $\now+1 \, \event \, \Qwithat' \tostate \withat{\tt end}$ 
that makes the contract transit to {\tt end} if it is still in $\Q'$ at the
beginning of the
next time-clock.

The encoding of $\Q : \DecJump(\R_i, \Q', \Q'')$ is more convoluted because we have to move all the
events $\now + 1 \event \withat{\tt dec}_i \tostate \withat{\tt ackdec}_i$ ahead two clock units
(except one, if the corresponding register value is positive). Assume an invocation of {\tt fdecQ} occurs
and $i = 1$. Then a bunch of events are created (see the body of {\tt fdecQ} in Table~\ref{tab.encodingTA})
and the contract transits into {\tt wait}. In this state, 
a \rulename{Tick} transition can occur; hence the time values of the events are decreased by one. Then $0 \event {\tt wait} 
\tostate {\tt dec}_1$ is enabled and the protocol moving ahead all the events 
$0 \event {\tt dec}_1 \tostate {\tt ackdec}_1$ (except one)
and
$0 \event {\tt dec}_2 \tostate {\tt ackdec}_2$ starts. The protocol works as follows:
\begin{enumerate}
\item
since the state is ${\tt dec}_1$, $0 \event {\tt dec}_1 \tostate {\tt ackdec}_1$ is fired
(assume the value of $R_1$ is positive) and the contract transits to  ${\tt ackdec}_1$;
\item
in ${\tt ackdec}_1$, $0 \event {\tt ackdec}_1 \tostate {\tt nextQ''}_1$ is fired
and the contract transits to state ${\tt nextQ''}_1$ (one event $0 \event {\tt dec}_1 \tostate {\tt ackdec}_1$
has been erased without being moved ahead);
\item
in ${\tt nextQ''}_1$, the function 
\[
{\withat{\tt nextQ''}}_1 \; \;  {\tt fnextQ''\_1} \; \{ \; \now + 1 \event \withat{\tt next} \tostate
\withat{\tt Q''} \} \; \tostate \; \withat{\tt dec}_1
\]
can be invoked. 
A transition to ${\tt dec}_1$ happens and the event 
$1 \event {\tt next} \tostate \Q''$ is created. When the transfer protocol terminates, this 
event will make the contract transit to $\Q''$;
\item
at this stage the transfer of events $0 \event {\tt dec}_i \tostate {\tt ackdec}_i$
occurs. At first
the protocol moves the events $0 \event {\tt dec}_1 \tostate {\tt ackdec}_1$ (every
such event is fired and then the function ${\tt fackdec}_1$ that recreates the same event at $\now + 2$ is executed)
then the function ${\tt fdec}_1$ is invoked and the same protocol is applied to the events 
$0 \event {\tt dec}_2 \tostate {\tt ackdec}_2$;
\item
at the end of the transfers, the function ${\tt fdec}_2$ is invoked and the contract transits to 
{\tt next};
\item
in {\tt next}, no event nor function can be executed, therefore a \rulename{Tick} occurs and
the time values of the events are decreased by one (in particular those
$2 \event {\tt dec}_i \tostate {\tt ackdec}_i$ that were transferred at step 4). 
Then $0 \event {\tt next} \tostate
\Q''$ that was created at step 3 is executed and the contract transits to $\Q''$.
\end{enumerate}
When the register $R_1$ is 0 and {\tt fdecQ} is invoked then the event at step 2 cannot be produced and 
the computation is fated to reach an {\tt end} state (by {\tt fdec1}, {\tt fdec2} or by \rulename{tick} and 
then performing $0 \event {\tt dec}_1 \tostate {\tt end}$). If, on the contrary, the invoked function is
{\tt fzeroQ}, the same protocol as above is used to transfer 
the events $0 \event {\tt dec}_i \tostate {\tt ackdec}_i$ (in this case with $i = 2$ only) and the
contract reaches the step 5 where, after a \rulename{Tick}, the event 
$\now+2 \, \event \, {\tt next} \tostate \Q'$ can be executed. The undecidability result for {\miniStipulaTA} follows.
\Reviewer{2}{Because the encoding of the Minsky machine works in a two phase scenario, it could be useful to have a result expressing the correspondence between transitions of the Minsky machine and transitions of its encoding. Such a result would likely express some kind of simulation property. Would this be possible, or are there technical aspects making this difficult?}
\Answer{This is exactly what we do in the proofs, which are longer and longer because the management protocols 
are increasingly complex and we need to analyze all the possible cases.}

\begin{restatable}{theorem}{thmtimeahead}
\label{thm.timeahead}
State reachability is undecidable in {\miniStipulaTA}.
\end{restatable}

\subsection{Undecidability results for {\miniStipulaD}}
\label{sec:determinate}

\Reviewer{2}{
Here are two remarks, possibly naive, about the definitions of the
subsets of microStipula. First, the SD fragment is quite poor. This is probably why the encoding of Minsky machines needs more infrastructure, and why defining an encoding is more difficult than in the other fragments. So all in all, does the SD fragment make sense from the point of view of "programming in Stipula", beyond the ability to prove undecidability of clause reachability?
}
\Answer{We have already answered in the Introduction to a similar remark. The contract {\tt alt\_BikeRental}
in the appendix conforms to the {\miniStipulaDI} constraints}

Also in this case we reduce 
from the halting problem of Minsky machines to 
state reachability in {\miniStipulaD}. 
%
In {\miniStipulaD}, functions and events start in different states. Therefore the encoding 
of Table~\ref{tab.encodingTA} is inadequate since we used the 
expedient that events preempt functions when enabled in the same state to make the contract transit 
to the {\tt end} state (which indicates an error).
For {\miniStipulaD} we need to refine the sequence \emph{machine-{\control} states} in order to 
have extra {\control} over erroneous operations (decrease of zero register or zero-test of a positive
register). The idea is to manage at different times the events ${\tt dec}_1 \tostate {\tt ackdec}_1$ and 
${\tt dec}_2 \tostate {\tt ackdec}_2$ that model registers' units. In particular, if the contract is in a 
(machine) state $\Q$  at time 0 then ${\tt dec}_1 \tostate {\tt ackdec}_1$ are at time 1 and 
${\tt dec}_2 \tostate {\tt ackdec}_2$ are at time 3. At times 2 and 4 {\control} states 
perform management
operations.
Therefore the sequence of states becomes\\
\indent\
$\mbox{\emph{machine-state} $\rightarrow$ \emph{transfer1-state} $\rightarrow$ \emph{{\control}1-state}
$\rightarrow$}$ \\
\indent
$\mbox{ \emph{transfer2-state} $\rightarrow$ \emph{{\control}2-state} 
}$

\noindent
where every state is one tick ahead the previous one. 
Therefore, \emph{transfer1-state} and \emph{transfer2-state} manage the transfer of 
${\tt dec}_1 \tostate {\tt ackdec}_1$ and 
${\tt dec}_2 \tostate {\tt ackdec}_2$ ahead five clock units.  
\begin{table}[p]
\hrulefill

{\small
Let $M$ be a Minsky machine with initial state $\Q_0$. Let {\tt D}$_M$ be the
{\miniStipulaD} contract
\begin{center}
{\tt stipula D$_M$ \{ 
        init Start  \quad
        $F_M$ 
    \}
} 
\end{center}
with $F_M$ contains the functions 
\begin{itemize}
\item
{\tt @Start fstart $\{$ \now $\event$ \withat{\tt notickA} $\tostate$ \withat{\tt cont} $\}$ $\tostate \; \Qwithat_0$  } 

\medskip

\item \ 
\vspace{-.6cm}
\[
\begin{array}{l@{\qquad\qquad}l}
\begin{array}{l}
\!\!\!\! \text{for every } \Q : \Inc(R_1, \Q'):
\\[.2cm]
\Qwithat: {\tt fAincQ} \; \{
\\
\qquad \now +1 \event \withat{\tt dec}_1 \tostate \withat{\tt ackdec}_1
\\
\qquad \now \event \withat{\tt cont} \tostate \Qwithat'
\\
\qquad \now \event \withat{\tt notickB}  \tostate \withat{\tt cont}
\\
\} \tostate \withat{\tt notickA}
\\[.2cm]
\Qwithat: {\tt fBincQ} \; \{
\\
\qquad \now +1 \event \withat{\tt dec}_1 \tostate \withat{\tt ackdec}_1
\\
\qquad \now \event \withat{\tt cont} \tostate \Qwithat'
\\
\qquad \now \event \withat{\tt notickA} \tostate \withat{\tt cont}
\\
\} \tostate \withat{\tt notickB}
\end{array}
&
\begin{array}{l}
\!\!\!\! \text{for every } \Q : \Inc(R_2, \Q'):
\\[.2cm]
\Qwithat: {\tt fAincQ} \; \{
\\
\qquad \now +3 \event \withat{\tt dec}_2 \tostate \withat{\tt ackdec}_2
\\
\qquad \now \event \withat{\tt cont} \tostate \Qwithat'
\\
\qquad \now \event \withat{\tt notickB} \tostate \withat{\tt cont}
\\
\} \tostate \withat{\tt notickA}
\\[.2cm]
\Qwithat: {\tt fBincQ} \; \{
\\
\qquad \now +3 \event \withat{\tt dec}_2 \tostate \withat{\tt ackdec}_2
\\
\qquad \now \event \withat{\tt cont} \tostate \Qwithat'
\\
\qquad \now \event \withat{\tt notickA} \tostate \withat{\tt cont}
\\
\} \tostate \withat{\tt notickB}
\end{array}
\end{array}
\]
\item \ 
\vspace{-.6cm}
\[
\begin{array}{l@{\quad}l}
\begin{array}{l}
\!\!\!\! \text{for every } \Q : \DecJump(R_1, \Q', \Q''):
\\[.2cm]
\Qwithat: {\tt fAdecQ} \; \{
\\
\qquad \now +1 \event \withat{\tt ackdec}_1 \tostate \Qwithat''\_{\tt start1}
\\
\qquad \now +1 \event \withat{\tt s1notick} \tostate {\tt cont}
\\
\qquad \now \event \withat{\tt cont} \tostate \withat{\tt dec}_1
\\
\}  \tostate \withat{\tt notickA}
\\
\\
\\
\\[.2cm]
\Qwithat: {\tt fBdecQ} \; \{
\\
\qquad \now +1 \event \withat{\tt ackdec}_1 \tostate \Qwithat''\_{\tt start1}
\\
\qquad \now +1 \event \withat{\tt s1notick} \tostate \withat{\tt cont}
\\
\qquad \now \event \withat{\tt cont} \tostate \withat{\tt dec}_1
\\
\} \tostate \withat{\tt notickB}
\\
\\
\\
\\[.2cm]
\Qwithat: {\tt fAzeroQ} \; \{
\\
\qquad \now \event \withat{\tt cont} \tostate \withat{\tt dec}_1
\\
\qquad \now +2 \event \withat{\tt dec}_1 \tostate \withat{\tt dec}_2
\\
\qquad \now +3 \event \withat{\tt ackdec}_2 \tostate \withat{\tt copy}_2
\\
\qquad \now +3 \event \withat{\tt c2notickA} \tostate \withat{\tt cont}
\\
\qquad  \now +4 \event \withat{\tt dec_2} \tostate \Qwithat'
\\
\qquad  \now +5 \event \withat{\tt notickB} \tostate \withat{\tt cont}
\\
\} \tostate \withat{\tt notickA}
\\[.2cm]
\Qwithat: {\tt fBzeroQ} \; \{
\\
\qquad       \now \event \withat{\tt cont} \tostate \withat{\tt dec}_1
\\
\qquad       \now +2 \event \withat{\tt dec_1} \tostate \withat{\tt dec}_2
\\
\qquad       \now +3 \event \withat{\tt ackdec_2} \tostate \withat{\tt copy}_2
\\
\qquad       \now +3 \event \withat{\tt c2notickA} \tostate \withat{\tt cont}
\\
\qquad       \now +4 \event \withat{\tt dec_2} \tostate \Qwithat'
\\
\qquad       \now +5 \event \withat{\tt notickA} \tostate \withat{\tt cont}
\\
\} \tostate \withat{\tt notickB}
\end{array}
&
\begin{array}{l}
\!\!\!\! \text{for every } \Q : \DecJump(R_2, \Q', \Q''):
\\[.2cm]
\Qwithat: {\tt fAdecQ} \; \{
\\
\qquad \now \event \withat{\tt cont} \tostate \withat{\tt dec}_1
\\
\qquad \now +1 \event \withat{\tt ackdec}_1 \tostate \withat{\tt copy}_1
\\
\qquad \now +1 \event \withat{\tt c1notickA} \tostate {\tt cont}
\\
\qquad \now +2 \event \withat{\tt dec}_1 \tostate {\tt dec}_2
\\
\qquad \now +3 \event \withat{\tt ackdec}_2 \tostate \Qwithat''{\tt \_start2}
\\
\qquad \now +3 \event \withat{\tt s2notick} \tostate \withat{\tt cont}
\\
\}  \tostate \withat{\tt notickA}
\\[.2cm]
\Qwithat: {\tt fBdecQ} \; \{
\\
\qquad \now \event \withat{\tt cont} \tostate \withat{\tt dec}_1
\\
\qquad \now +1 \event \withat{\tt ackdec}_1 \tostate \withat{\tt copy}_1
\\
\qquad \now +1 \event \withat{\tt c1notickA} \tostate {\tt cont}
\\
\qquad \now +2 \event \withat{\tt dec}_1 \tostate {\tt dec}_2
\\
\qquad \now +3 \event \withat{\tt ackdec}_2 \tostate \Qwithat''{\tt \_start2}
\\
\qquad \now +3 \event \withat{\tt s2notick} \tostate \withat{\tt cont}
\\
\}  \tostate \withat{\tt notickB}
\\[.2cm]
\Qwithat: {\tt fAzeroQ} \; \{
\\
\qquad \now \event \withat{\tt cont} \tostate \withat{\tt dec}_1
\\
\qquad \now +1 \event \withat{\tt ackdec}_1 \tostate \withat{\tt copy}_1
\\
\qquad \now +1 \event \withat{\tt c1notickA} \tostate \withat{\tt cont}
\\
\qquad \now +2 \event \withat{\tt dec}_1 \tostate \withat{\tt dec}_2
\\
\qquad  \now +4 \event \withat{\tt dec}_2 \tostate \Qwithat'
\\
\qquad  \now +5 \event \withat{\tt notickB} \tostate \withat{\tt cont}
\\
\} \tostate \withat{\tt notickA}
\\[.2cm]
\Qwithat: {\tt fBzeroQ} \; \{
\\
\qquad       \now  \event \withat{\tt cont} \tostate \withat{\tt dec}_1
\\
\qquad       \now +1 \event \withat{\tt ackdec_1} \tostate \withat{\tt copy}_1
\\
\qquad       \now +1 \event \withat{\tt c1notickA} \tostate \withat{\tt cont}
\\
\qquad       \now +2 \event \withat{\tt dec}_1 \tostate \withat{\tt dec}_2
\\
\qquad       \now +4 \event \withat{\tt dec}_2 \tostate \Qwithat'
\\
\qquad       \now + 5\event \withat{\tt notickA} \tostate \withat{\tt cont}
\\
\} \tostate \withat{\tt notickB}
\end{array}
\end{array}
\]
\item
the management functions in Table~\ref{tab.encodingD_addendum}.
\end{itemize}
}
\hrulefill
\caption{\label{tab.encodingD} The {\miniStipulaD} contract modelling a Minsky machine}
\end{table}
\begin{table}[t]
\hrulefill
{\small
\[
\begin{array}{l@{\qquad}l}
\begin{array}{l}
\withat{\tt Q\_start1}:  {\tt fQstart1} \; \{ 
\\
\qquad \now    \event \withat{\tt ackdec_1} \tostate \withat{\tt copy}_1
\\ \qquad \now    \event \withat{\tt cont} \tostate \withat{\tt dec}_1
\\ \qquad \now    \event \withat{\tt c1notickA} \tostate \withat{\tt cont}
\\ \qquad \now +1 \event \withat{\tt dec}_1 \tostate \withat{\tt dec}_2
\\ \qquad \now +2 \event \withat{\tt ackdec}_2 \tostate \withat{\tt copy}_2
\\ \qquad \now +2 \event \withat{\tt c2notickA} \tostate \withat{\tt cont}
\\ \qquad \now +3 \event \withat{\tt dec}_2 \tostate \Qwithat
\\ \qquad \now +4 \event \withat{\tt notickA} \tostate \withat{\tt cont}
\\
\} \tostate \withat{\tt s1notick}
\\[.2cm]
\withat{\tt copy}_1: {\tt fAcopy1} \; \{
\\ \qquad \now  \event \withat{\tt ackdec}_1 \tostate \withat{\tt copy}_1
\\ \qquad \now    \event \withat{\tt cont} \tostate \withat{\tt dec}_1
\\ \qquad \now    \event \withat{\tt c1notickB} \tostate \withat{\tt cont}
\\ \qquad \now +5 \event \withat{\tt dec}_1 \tostate \withat{\tt ackdec}_1
\\
\} \tostate \withat{\tt c1notickA}
\\[.2cm]
\withat{\tt copy}_2: {\tt fAcopy2} \; \{
\\ \qquad \now  \event \withat{\tt ackdec}_2 \tostate \withat{\tt copy}_2
\\ \qquad \now    \event \withat{\tt cont} \tostate \withat{\tt dec}_2
\\ \qquad \now    \event \withat{\tt c2notickB} \tostate \withat{\tt cont}
\\ \qquad \now +5 \event \withat{\tt dec}_2 \tostate \withat{\tt ackdec}_2
\\
\} \tostate \withat{\tt c2notickA}
\end{array}
&
\begin{array}{l}
\withat{\tt Q\_start2}: {\tt fQstart2} \; \{
\\ \qquad \now    \event \withat{\tt ackdec}_2 \tostate \withat{\tt copy}_2
\\ \qquad \now    \event \withat{\tt cont} \tostate \withat{\tt dec}_2
\\ \qquad \now    \event \withat{\tt c2notickA} \tostate \withat{\tt cont}
\\ \qquad \now +1 \event \withat{\tt dec}_2 \tostate \Qwithat
\\ \qquad \now +2 \event \withat{\tt notickA} \tostate \withat{\tt cont}
\\
\} \tostate \withat{\tt s2notick}
\\
\\
\\
\\[.2cm]
\withat{\tt copy}_1: {\tt fBcopy1} \; \{
\\ \qquad \now  \event \withat{\tt ackdec}_1 \tostate \withat{\tt copy}_1
\\ \qquad \now    \event \withat{\tt cont} \tostate \withat{\tt dec}_1
\\ \qquad \now    \event \withat{\tt c1notickA} \tostate \withat{\tt cont}
\\ \qquad \now +5 \event \withat{\tt dec}_1 \tostate \withat{\tt ackdec}_1
\\
\} \tostate \withat{\tt c1notickB}
\\[.2cm]
\withat{\tt copy}_2: {\tt fBcopy2} \; \{
\\ \qquad \now  \event \withat{\tt ackdec}_2 \tostate \withat{\tt copy}_2
\\ \qquad \now    \event \withat{\tt cont} \tostate \withat{\tt dec}_2
\\ \qquad \now    \event \withat{\tt c2notickA} \tostate \withat{\tt cont}
\\ \qquad \now +5 \event \withat{\tt dec}_2 \tostate \withat{\tt ackdec}_2
\\
\} \tostate \withat{\tt c2notickB}
\end{array}
\end{array}
\]
} 
\hrulefill

\caption{\label{tab.encodingD_addendum} The management functions of Table~\ref{tab.encodingD} ($\Q$ is 
a state of the Minsky machine)}
\end{table}

Clearly,  misplaced \rulename{Tick} transitions may break the rigidity of the protocol. 
This means that it is necessary to stop the simulation if a wrong \rulename{Tick} 
transition is performed. We already used a similar mechanism in Table~\ref{tab.encodingI}.
In that case, a {\control}
event at time 0 (that is created in the past transition) is necessary to simulate the 
Minsky machine transition; in turn, the simulation creates a similar {\control} event for the next one.
Therefore, if a tick occurs before the invocation of a function (thus erasing registers' values that were
events at time 0, as well) then the simulation stops because the {\control} event is also erased.

A similar expedient cannot be used for the {\miniStipulaD} encoding because the registers' values
are at different times (+1 and +3 with respect to the machine state) and erasing the {\control} event 
with a tick may be useless if the {\miniStipulaD} function produces an equal {\control} event, which is 
the case when the corresponding transition is circular (initial and final states are the same).
Therefore we refine the technique in Table~\ref{tab.encodingI} by adding sibling functions and the
simulation uses standard functions or sibling ones according to the presence of the {\control} event 
$0 \event {\tt notickA} \tostate {\tt cont}$ or of $0 \event {\tt notickB} \tostate {\tt cont}$. For
example, the encoding of $\Q : \Inc(R_1, \Q')$ is (the events of register $R1$ are at $\now + 1$):
\[
\begin{array}{l@{\qquad\qquad}l}
\begin{array}{l}
\Qwithat: {\tt fAincQ} \; \{
\\
\qquad \now +1 \event \withat{\tt dec}_1 \tostate \withat{\tt ackdec}_1
\\
\qquad \now \event \withat{\tt cont} \tostate \Qwithat'
\\
\qquad \now \event \withat{\tt notickB}  \tostate \withat{\tt cont}
\\
\} \tostate \withat{\tt notickA}
\end{array}
&
\begin{array}{l}
\Qwithat: {\tt fBincQ} \; \{
\\
\qquad \now +1 \event \withat{\tt dec}_1 \tostate \withat{\tt ackdec}_1
\\
\qquad \now \event \withat{\tt cont} \tostate \Qwithat'
\\
\qquad \now \event \withat{\tt notickA} \tostate \withat{\tt cont}
\\
\} \tostate \withat{\tt notickB}
\end{array}
\end{array}
\]
Assuming to be in a configuration with state $\Q$ and event $0 \event {\tt notickA} \tostate {\tt cont}$,
the unique function that can be invoked is {\tt fAincQ}; thereafter, the combined effect of 
$0 \event {\tt notickA} \tostate {\tt cont}$ and $0 \event {\tt cont} \tostate \Q'$ allows the
contract to transit to $\Q'$ with an additional event 
$1 \event {\tt dec}_1 \tostate {\tt ackdec}_1$ (corresponding to a register increment) and 
the presence of $0 \event {\tt notickB} \tostate {\tt cont}$ that compels the next instruction,
if any, to be a sibling one (\emph{e.g.}~${\tt fBincQ}'$).

Tables~\ref{tab.encodingD} and~\ref{tab.encodingD_addendum} define the encoding of a Minsky machine $M$ into a 
{\miniStipulaD} contract ${\tt D}_M$. The reader may notice that the management functions {\tt fQstart1}
and {\tt fQstart2} in 
Table~\ref{tab.encodingD_addendum} do not have sibling functions -- they always produce a {\control} event 
${\tt k} \event {\tt notickA} \tostate {\tt cont}$ (with {\tt k} be either 2 or 4). Actually this is an optimization: they are invoked  
because a decrement occurred and, in such cases it is not possible that the foregoing {\control} events 
are used in the simulation of the current instruction. The undecidability result for {\miniStipulaD} follows.

\Reviewer{2}{Because of the intricacies in the encoding, one wonders whether a paper proof is enough. Some kind of mechanised proof would be welcome to convince the reader. Alternatively, it would be helpful to break down reasoning about the semantics of the contract into several steps, thus increasing modularity. Could you comment on these two potential routes? Which one would you take to increase confidence in the result about SD?
}
\Answer{The proofs, theoretically, are not so difficult and all use the same technique:
defining a simulation relation and demonstrating its correctness. Since the protocol
of {\miniStipulaD} is entangled, the proof is much longer because it has to analyze all the
cases. We have written the cases, but it takes time to 
write it down in latex. As regards the automatic proof, it is a good idea but, honestly,
we do not have the competences.}

\begin{restatable}{theorem}{thmdeterminate}
\label{thm.determinate}
State reachability is undecidable in {\miniStipulaD}.
\end{restatable}

\section{Decidability results for {\miniStipulaDI}}
\label{sec:det_instantaneous}

%
\Reviewer{1}{the definition of coverability for well-structured transition systems (WSTSs) for the decidability result is flawed, raising doubts. Also, the used terminology is highly misleading. The authors do actually consider the clause reachability problem and not a general reachability problem for configurations of the underlying transition system. For instance, the latter might not be decidable for WSTS but the coverability problem is. It might be that the main definition (Def. 1) should not say "does not contain" but "does contain". }
\Answer{repaired}

We demonstrate that 
state reachability is decidable for {\miniStipulaDI} by reasoning on a variant  with an alternative 
\rulename{Tick} rule. 
We recall that, in {\miniStipulaDI}, for every $\Q \; \f \; \Q', \; \Q'' \; \ev \; \Q'''  \in \C$, 
we have $\Q \neq \Q''$

Let $\InitEv(\C)$ be the set of initial states of
events in $\C$, where $\C$ is a {\miniStipula} contract. Let 
{\small\[
\mathrule{Tick-Plus}{
	\Q \notin \InitEv(\C)
	}{
	\pairbis{\C(\Q \semi \zero \semi \SemEvent)}{\time} \lred{} 
	\pairbis{\C(\Q  \semi \zero \semi \SemEvent \downarrow)}{\time + 1}
	}
\]}

\medskip

\noindent
That is, unlike \rulename{Tick}, \rulename{Tick-Plus} may only be used in states that are not initial
states of events. Let {\miniStipulaDIP} be the language whose operational semantics uses \rulename{Tick-Plus}
instead of \rulename{Tick}; we denote with $\lredTickP{}$ the transition relation 
of {\miniStipulaDIP}. We observe that, syntactically, nothing is changed: every {\miniStipulaDI} contract is a {\miniStipulaDIP} contract
and conversely. 
We denote by $\TStp(\C)=(\mathcal{C}_\C,\lredTickP{})$ the transitions system associated to contract $\C$ using $\lredTickP{}$ as transition rule.

\Reviewer{2}{
Would it make sense to rely on Proposition 1 for the undecidability results in Section 3? Would this simplify some proofs?}
\Answer{We already rely on it: see Corollary 1}

\begin{definition}
Let $\contract$ be a possible configuration of a {\miniStipula} contract. We say that $\contract$ is \emph{stuck}
if, for every computation $\contract \lred{}^* \contract'$ the transitions therein are always instances of 
\rulename{Tick}.
\end{definition}

\begin{restatable}{proposition}{propcorrespondence}
\label{prop.correspondence} 
Let $\pairbis{\C(\Q \semi \Sigma \semi \Psi)}{\time}$ be a configuration of a {\miniStipulaDI} contract 
$\C$ (or a {\miniStipulaDIP} contract).
Then 
\begin{enumerate}
\item[(i)] whenever $\Q \notin \InitEv(\C)$ or $\Sigma \neq \zero$:
\[
\pairbis{\C(\Q \semi \Sigma \semi \Psi)}{\time} \lred{} \contract \quad \text{ if and only if } \quad
\pairbis{\C(\Q \semi \Sigma \semi \Psi)}{\time} \lredTickP{} \contract \; ;
\]
\item[(ii)]
whenever $\Q \in \InitEv(\C)$ and $\Psi = {\tt 0} \event_{\! {\tt n}} \Q \tostate
\Q' ~|~ \Psi'$:
\[
\pairbis{\C(\Q \semi \zero \semi \Psi)}{\time} \lred{} \contract
 \quad \text{ if and only if } \quad
\pairbis{\C(\Q \semi \zero \semi \Psi)}{\time} \lredTickP{} \contract \; ;
\]
\item[(iii)]
whenever $\Q \in \InitEv(\C)$  and $\noredbis{\SemEvent, \Q}$ :
\[
\pairbis{\C(\Q \semi \zero \semi \Psi)}{\time} \; \text{ is stuck}
 \quad \text{ if and only if } \quad
\pairbis{\C(\Q \semi \zero \semi \Psi)}{\time} \nolredTickP{} .
\]

\end{enumerate}
\end{restatable}

A consequence of Proposition~\ref{prop.correspondence} is that a state $\Q$ 
is reachable in {\miniStipulaDI} if and only if it is reachable in {\miniStipulaDIP}.
This allows us to safely reduce state reachability arguments to {\miniStipulaDIP}.
In particular we demonstrate that $\TStp(\C)$, where $\C$ is a  {\miniStipulaDIP} contract, is a well-structured transition system.
 
We begin with some background on well-structured transition systems~\cite{Finkel:2001}.
A relation $\le \subseteq X \times X$ is called \emph{quasi-ordering} if it is 
reflexive and transitive.
A \emph{well-quasi-ordering} is a quasi-ordering $\le \subseteq X \times X$ such that, for every infinite sequence $x_1, x_2, x_3, \cdots$, there exist $i < j$ with $x_i \leq x_j$.

\begin{definition}
\label{def.wsts}
A \emph{well-structured transition system} 
is a tuple $(\mathcal{C}, \lred{}, \preceq)$ where $(\mathcal{C}, \lred{})$ is 
a transition system and $\preceq \subseteq \mathcal{C} \times \mathcal{C}$ is 
a quasi-ordering such that:
\begin{enumerate}
\item[(1)] 
$\preceq$ is a well-quasi-ordering
\item[(2)] 
$\preceq$ is upward compatible with $\lred{}$, i.e., for every 
$\contract_1,\contract'_1,\contract_2 \in \mathcal{C}$ such that $\contract_1 
\preceq \contract_1'$ and $\contract_1 \lred{} \contract_2$ there exists 
$\contract_2'$ in $\mathcal{C}$ verifying $\contract_1' 
\lred{}^* \contract_2'$ and $\contract_2 \preceq \contract_2'$
\end{enumerate}
\end{definition}

\Reviewer{1}{clash in symbol use for set of configurations}
\Answer{Repaired}

Given a configuration $\contract$ of a well-structured transition system, 
$\Pred(\contract)$ denotes the set of immediate predecessors of $\contract$ (\emph{i.e.}, 
$\Pred(\contract) = \{ \contract' \; | \; \contract' \lred{} \contract \; \}$) 
while $\uparrow \contract$ denotes the set of configurations greater than $\contract$ (\emph{i.e.},
$\uparrow \contract = \{ \; \contract' \; | \; \contract \preceq \contract' \; \}$). 
A \emph{basis} of an upward-closed set of configurations 
$\mathcal{D} \subseteq \mathcal{C}$ is a set $\mathcal{D}^\flat$
such that $\mathcal{D} = \cup_{\contract \in \mathcal{D}^\flat} \uparrow \! \contract$. 
We know that every upward-closed set of a well-quasi-ordering 
admits a finite basis~\cite{Finkel:2001}.
With abuse of notation, we will denote with $\Pred( \cdot )$ also its natural extension to sets of 
configurations.

Several properties are decidable for well-structured transition systems (under some conditions discussed below)~\cite{Abdulla1996,Finkel:2001}, we
 will consider the following one.

\Reviewer{3}{the definition of the coverability problem (p. 16) seems wrong. The correct definition of the coverability problem should have the form:
given $c,d$, decide whether there exists $d'$ such that $c \longrightarrow^* d' \succeq 
 d$
}
\Answer{Repaired}

\Reviewer{3}{
the definition of $\longrightarrow^*$ is ambiguous. Depending on whether it is possible to have h=0 or not, the relation is reflexive (as it usually is) or not. Assuming the usual reflexive relation, however, would make Definition 4 trivial (however, this definition is probably wrong, see the previous point). Furthermore, if $\longrightarrow^*$ is reflexive, then the argument given after Lemma 1 seems false, as to prove upward compatibility it should be enough to take $c = c' = (Q,\_,0 \event_n Q \tostate Q')$. }
\Answer{Repaired}

\begin{definition}
Let $(\mathcal{C}, \lred{}, \preceq)$ be  a well-structured transition system.
The \emph{coverability problem} is to decide, given the initial configuration 
$\contractin\in \mathcal{C}$ 
and a target configuration $\contract \in \mathcal{C}$, 
whether there exists a configuration $\contract' \in \mathcal{C}$ such that 
$\contract \preceq \contract'$ and  $\contractin \lred{}^* \contract'$.
\end{definition}
In well-structured transition systems the coverability problem is decidable when the transition 
relation $\lred{}$, the ordering $\preceq$ and a finite-basis for the set of configurations $\Pred(\uparrow \contract)$ 
are effectively computable. 

Let us now define the relation $\preceq$  as the least quasi-ordering relation such that $\Psi \preceq \Psi ~|~ \Psi'$ for every $\Psi'$.
The relation $\preceq$ is lifted to configurations as follows
\[
\pairbis{\C(\Q \semi \Sigma \semi \Psi)}{\time} \preceq \pairbis{\C(\Q \semi \Sigma \semi \Psi')}{\time}
\quad \text{if} \quad 
\Psi \preceq \Psi'\; .
\]
\Reviewer{3}{
the paragraph ``It is worth to observe [...]'' is unclear: please check carefully and rephrase
}
\Answer{Repaired}
It is worth to observe that, according to the relation $\preceq$, 
the state reachability problem for $\Q$ is equivalent to the coverability problem 
for the initial configuration $\contractin$ and  the target configuration 
$\C(\Q, \zero, \zero)$.

\Reviewer{1}{Lm.1: "Its transition system" is not properly defined}
\Answer{Repaired}

\begin{restatable}{lemma}{lemwellstructuredTP}
\label{lem.wellstructuredTP}
For a {\miniStipulaDIP} contract $\C$, $(\mathcal{C}_\C, \lredTickP{}, \preceq)$  is a well-structured transition system.
\end{restatable}

It is worth to notice that Lemma~\ref{lem.wellstructuredTP} does not hold for {\miniStipulaDI} because $\lred{}$
is not upward compatible with $\preceq$. In fact, while $\C(\Q, \zero, \zero) \lred{} \C(\Q, \zero, \zero)$
with a rule \rulename{Tick} and $\C(\Q, \zero, \zero) \preceq \C(\Q, \zero, 0 \event_{\tt n} \Q \tostate \Q')$, 
the unique computation of $\C(\Q, \zero, 0 \event_{\tt n} \Q \tostate \Q')$ when $\Q' \in \InitEv(\C)$ is
(we recall that, in {\miniStipula}, events preempt function invocations and ticks):
\[
\C(\Q, \zero, 0 \event_{\tt n} \Q \tostate \Q') \lred{} 
\C(\Q, \zero \tostate \Q', \zero) \lred{} \C(\Q', \zero, \zero)
\]
and then it gets stuck. Therefore no configuration $\contract$
is reachable such that 
$\C(\Q, \zero, \zero) \preceq \contract$.

\Reviewer{1}{
The formalisation for this result is not quite solid. Most notably, the definition of coverability is always trivially satisfied by reflexivity - in fact, usually one requires an initial state from which a target state shall be covered - here, only one is considered. Also, for instance, the transition system of a contract is technically not defined and Lemma 1 is very liberal with the use of the ordering: first for multisets, then for contracts/configurations of contracts (which is also quite liberal). 
}
\Answer{The notion of coverability has been repaired, the transition system of a contract has been formally defined. We think that, now, the Lemma 1 is clearer.}

\begin{restatable}{lemma}{lemfinitepredbasis}
\label{lem.finitepredbasis}
%
Let $(\mathcal{C}, \lredTickP{}, \preceq)$ be the well-structured transition system of a {\miniStipulaDIP} contract. Then, $\lredTickP{}$ and
$\preceq$ are decidable and there exists an algorithm  for computing a finite basis of 
$\Pred(\uparrow \mathcal{D})$ for any finite $\mathcal{D} \subseteq \mathcal{C}$.
\end{restatable}

From Lemma~\ref{lem.finitepredbasis} and the above mentioned results, on well-structured transition systems 
we get the following result.
\begin{theorem}
\label{thm.decidableDIP}
The state reachability problem is decidable in {\miniStipulaDIP}.
\end{theorem}

As a direct consequence of Proposition~\ref{prop.correspondence}  and Theorem~\ref{thm.decidableDIP} we have:

\begin{corollary}
The state reachability problem is decidable in {\miniStipulaDI}.
\end{corollary}

\section{Related work}
\label{sec:related}

The decidability of problems about infinite-state systems has been largely addressed 
in the literature. We refer 
to~\cite{Abdulla1996} for an overview of the research area. 

It turns out that critical features of {\miniStipula}, such as garbage-collecting elapsed events or
preempting events with respect to functions and progression of time may be modelled
by variants of Petri nets with inhibitor and reset arcs. While standard Petri nets, which are
infinite-state systems, have decidable problems of reachability, coverability, boundedness, etc.
(see, e.g.,~\cite{EsparzaN94}), the above variants of Petri nets are Turing complete, therefore
all non-trivial properties become undecidable~\cite{DufourdFS98}. 
%
%
%

It is not obvious whether Petri nets with inhibitor and reset arcs may be modelled by {\miniStipula}
contracts. It seems that these features have different expressive powers than time progression and events.
Hence, other formalisms, such as pi-calculus and actor languages, might have a stricter correspondence 
with {\miniStipula}. As regards the decidability of problems in pi-calculus, we recall 
the decidability of reachability and
termination in the depth-bounded 
fragment of the pi-calculus~\cite{meyer08} and
the decidability of the reachability problem for various fragments of the asynchronous pi-calculus 
that feature name generation, name mobility, and unbounded control~\cite{Amadio2002}.
Regarding actor languages, in~\cite{Boer2014}, we demonstrated the decidability of termination 
for stateless actors 
(actors without fields) and for actors with states when the number of actors is bounded and the state 
is read-only. It is worth to observe that all these results have been achieved by using techniques 
that are similar to those used
in this paper: either demonstrating that the model of the calculus 
is a well-structured transition system~\cite{Finkel:2001} (for which, under certain computability conditions,
the reachability and termination problems are decidable, see Section~\ref{sec:det_instantaneous})
or simulating a Turing complete model, such as the Minsky machines, into the calculus under analysis
(hence the undecidability results of problems such as termination).

The {\miniStipula} calculus has some similarities with formal models of timed systems such as timed automata \cite{AD94}. The control state reachability problem, namely, given a timed automaton $A$ and a control state $q$, does there exist a run of $A$ that visits $q$, is known to be decidable \cite{AD94}. A similar result holds for Timed Networks (TN) \cite{abdulla-timed-01,abdulla-model-03}, a formal model consisting of a family of timed automata
with a distinct {\em controller} defined as a finite-state automaton without clocks.
Each process in a TN can communicate with all other processes via rendezvous messages. 
Control state reachability is also decidable for Timed Networks with transfer actions \cite{TAHN15}.
A transfer action forces all processes in a given state to move to a successor state as transfer arcs in Petri nets, which allows to move all tokens contained in a certain place to another \cite{finkel-monotonic-02}.
{\miniStipula} can also be seen as a language for modelling asynchronous programs in which callbacks are scheduled using timers. Verification problems for formal models of (untimed) asynchronous programs have been considered in \cite{SV06,JM07}.
In this context, Boolean program execution is modeled using a pushdown automaton, while asynchronous calls are modeled by adding a multiset of pending callbacks to the global state of the program. Callbacks are only executed when the program stack is empty.  Verification of safety properties is decidable for this model via a non-trivial reduction to coverability of Petri nets \cite{GM12}.

%

\section{Conclusions}
\label{sec:conclusions}

We have systematically studied the computational power of {\miniStipula}, a basic calculus defining 
legal contracts. The calculus is stateful and features clauses that may be either 
functions to be invoked by the external environment or
events that can be executed at certain time slots. We have demonstrated that in several 
legally relevant fragments of {\miniStipula} a problem such as
reachability of state is undecidable. The decidable fragment, {\miniStipulaDI},
 is the one whose event and functions start 
in disjoint states and where events are instantaneous (the time expressions are 0).

We conclude by indicating some relevant line
for future research. 
First of all, the decidability result for {\miniStipulaDI}
leaves open the question about the complexity of the
state reachability problem in that fragment. Our current conjecture is that
the problem is EXPSPACE-complete and we plan to prove this
conjecture by reducing the coverability problem for Petri nets
into the state reachability problem for {\miniStipulaDI}.
Another interesting line of research regards the investigation
of sound, but incomplete, algorithms for checking state reachability
in {\miniStipula}. A preliminary algorithm was investigated 
in~\cite{Laneve04}. The presented algorithm
spots clauses that are unreachable in {\miniStipula}
contracts and is
not tailored to any particular fragment of the language.
The results in this paper show that this algorithm 
may be improved
to achieve completeness when the input contract
complies with the {\miniStipulaDI} constraints.


\bibliography{bibliography}

\newpage

\appendix
\section{Legal relevance of {\miniStipula} and its fragments}
\label{sec:Section2}

To illustrate the relevance in a legal context of the fragments
of {\miniStipula} identified in the Introduction, 
we refer to simple legal contracts and employ a subset of operations from the full {\Stipula} language. 
We begin with a legal contract for bike rentals; it consists 
of three Articles:
\begin{enumerate}
\item
This Agreement shall commence when the Borrower takes possession of the Bike and remain in full 
force and
effect until the Bike is returned to Lender. The Borrower shall return the Bike {\tt k}
hours after the
rental and will pay Euro {\tt cost} in advance where half of the amount is of surcharge for late return.
\item
Payment.
Borrower shall pay the amount specified in Article 1 when this agreement commences. 
\item
End. If the Bike is returned within the agreed time specified in Article 1, then half of the 
{\tt cost} will be returned to Borrower; the other half is given to the Lender. Otherwise the
full {\tt cost} is given to the Lender.
\end{enumerate}
These clauses are transposed into {\miniStipula} using a few extensions to the calculus
-- operations that are almost standard and will be briefly discussed below.  


{\scriptsize
\begin{lstlisting}[language=C,mathescape=true,basicstyle=\ttfamily\scriptsize,firstnumber=1,numbers=left,numbersep=5pt,numberstyle=\tiny\color{teal},commentstyle=\color{olive},caption={The bike-rental contract},captionpos=b,label={lst:BikeRental}]
stipula Bike_Rental {  
  assets wallet, bike
  init @Inactive
  @Inactive Lender : offer[b] {
    b $\lolli$ bike  // the bike access code is stored in the contract
  } $\tostate$ @Payment
  @Payment Borrower : pay[x] 
    (x == cost) {
      x $\lolli$ wallet  // the contract keeps the money 
      bike $\lolli$ Borrower // the Borrower keeps the bike access code
      now + k $\event$ @Using {
          wallet $\lolli$ Lender // deadline expired: the whole wallet to Lender
      } $\tostate$ @End
  } $\tostate$ @Using
  @Using Borrower : end { // bike returned before the deadline
    0.5 $\times$ wallet $\lolli$ Lender // half wallet to Lender
    wallet $\lolli$ Borrower // half wallet to Borrower
  } $\tostate$ @End
}
\end{lstlisting}
}

The contract is defined by the keyword {\tt stipula} and is initially in the
state specified by the clause {\tt init}. This contract has two \emph{assets}: 
{\tt wallet} and {\tt bike} that 
will contain money and the access code of a bike, respectively. The contract starts when 
the Lender offers an access code {\tt b} of a bike (line 4) that is stored 
in the asset {\tt bike} -- the operation ${\tt b \lolli bike}$ at line 5 (this is Article 1 in the legal contract). Then the Borrower
can pay -- the function at line 7 -- which amounts to store in the {\tt wallet} the 
payment {\tt x} that has been made
(this is Article 2). At the same time the bike code is sent to the Borrower (the 
operation ${\tt bike \lolli Borrower}$ at line 10), so the Borrower can use the bike,
and event is scheduled at time {\tt now + k} 
(lines 11-13): if the bike is not returned at that time, the Borrower has to pay the
whole amount in the {\tt wallet} as specified in Article 3. The function at lines 15-18
defines the timely return of the bike by the Borrower (Article 3; notice that the
operation {\tt  0.5 $\times$ wallet $\lolli$ Lender} halves the content of the 
{\tt wallet}; therefore the following operation ${\tt wallet \lolli Borrower}$ sends to 
the Borrower the remaining half). 

A substantial problem in the
legal contract domain is spotting normative clauses, either  
functions or 
events, that will be never applied. In fact, it is possible, that an ureachable clause 
is considered too oppressive by one of the parties (take, for instance, 
the event of the foregoing Article 3), thus making the legal relationship fail.
To address this problem, we have defined {\miniStipula}, a sub-calculus of {\Stipula}
 that allows us to focus on the control flow of a contract. For example, you get a
{\miniStipula} contract by dropping all the operations $\lolli$, the parameters of
the functions and the asset declarations from Listing~\ref{lst:BikeRental}.
In particular, the resulting contract belongs to the fragment {\miniStipulaTA}
because {\tt k} (the renting time) is positive. 

There are other fragments of {\miniStipula} that are also relevant. Consider the 
following alternative bike rental contract, in which the Borrower may keep the 
bike indefinitely, but becomes eligible for a discount on the price if the bike is returned.
\begin{enumerate}
\item
This Agreement shall commence when the Borrower takes possession of the Bike and 
remain in full force and effect until the Bike is returned to Lender. 
The Borrower shall return the Bike at his discretion
and will pay Euro {\tt cost} in advance.
\item
Payment.
Borrower shall pay the amount {\tt cost} as specified in Article 1 when this agreement commences. 
\item
End. If the Bike is returned then 90\% of {\tt cost} will be given to Lender and 10\% of 
{\tt cost} will be returned to Borrower.
\item
Problems and resolutions. If the Lender detects a problem, he may trigger a resolution process that is managed by 
the Authority. The Authority will immediately enforce resolution to either the Lender or the Borrower.
\end{enumerate}
(To keep things simple, we have omitted the details of Article 4.) The alternative
contract is similar to that of Listing~\ref{lst:BikeRental} up-to line 7. 

{\scriptsize
\begin{lstlisting}[language=C,mathescape=true,basicstyle=\ttfamily\scriptsize,firstnumber=1,numbers=left,numbersep=5pt,numberstyle=\tiny\color{teal},commentstyle=\color{olive},caption={The alternative bike-rental contract},captionpos=b,label={lst:altBikeRental}]
stipula alt_Bike_Rental {  
  assets wallet, bike
  init @Inactive
  @Inactive Lender : offer[b] {
    b $\lolli$ bike  // the bike access code is stored in the contract
  } $\tostate$ @Payment
  @Payment Borrower : pay[x] 
    (x >= cost) {
      x $\lolli$ wallet  // the contract keeps the money 
      bike $\lolli$ Borrower // the Borrower keeps the bike access code
  } $\tostate$ @Using
  @Using Borrower : end { // bike returned 
    0.9 $\times$ wallet $\lolli$ Lender // 90$\%$ wallet to Lender
    wallet $\lolli$ Borrower // 10$\%$ wallet to Borrower
  } $\tostate$ @End
  @Using Lender : problems(u) { 
  	// communicate the problems u to Authority that will manage the issue
  } $\tostate$ @Manage
  @Manage Authority : resolution (v) { 
    if (v == 0) {
       now $\event$ @Problem {  
          // immediately enforce resolution to the Lender
       } $\tostate$ @PbLender 
    } else { // there is no management to do
       now $\event$ @Problem { // immediately enforce resolution to the Borrower
       } $\tostate$ @PbBorrower 
    } 
  } $\tostate$ @Problem
\end{lstlisting}
}

The key lines of Listing~\ref{lst:altBikeRental} are 21 and 25 where the Authority
enforces an immediate resolution in favour of the Lender of the Borrower. In this case, the events are now: it turns out that the contract belongs to the fragment {\miniStipulaI}.
Actually, it also belongs to the fragment {\miniStipulaD} because the initial 
state {\tt Problem} of the two events is different from {\tt Inactive}, {\tt Payment}, 
{\tt Using} and {\tt Manage}, which are the initial state of functions. As such, 
the contract of Listing~\ref{lst:altBikeRental} belongs to the fragment {\miniStipulaDI}.

\section{Proofs}
\label{sec:proofs}

\subsection{Proofs of Section \ref{sec:undecidability}}

\Reviewer{1}{In proofs, the notation \_ is confusing as \_ is often used as wildcard in the sense that the value does not matter.
}
\Answer{
The symbol \_ is overloaded but the context removes the ambiguities as defined in Section~\ref{sec:operationalsemantics}.
}
\thmintantaneous*

\begin{proof}
Let $M$ be a Minsky machine with states $\{\Q_0, \cdots , \Q_n\}$ where $\Q_0$ is the initial state.
Let ${\tt I}_M$ be the encoding in {\miniStipulaI}
as defined in Table~\ref{tab.encodingI} and let $\S, \S', \cdots$ range over states of ${\tt I}_M$.
After a sequence of transitions \rulename{Tick}, the contract ${\tt I}_M$ performs the
transitions
\[
{\tt I}_M({\tt Start}, \zero,\zero) \lred{{\tt fstart}} 
{\tt I}_M({\tt Start}, {\it Ev} \tostate \Q_0 ,\zero) \lred{} {\tt I}_M(\Q_0, \zero,{\it Ev})
\]
where ${\it Ev} = 0 \; \event \; {\tt aQ}_0 \tostate 
{\tt bQ}_0$. Next, let 
\[
\sem{v_1,v_2}_\Q^{\tt I} = \underbrace{0 \, \event \, \mathtt{dec}_1 \tostate {\tt ackdec}_1}_{v_1 \text{ times}}
\; | \; \underbrace{0 \, \event \, \mathtt{dec}_2 \tostate {\tt ackdec}_2}_{v_2 \text{ times}}
\; | \; 0 \, \event \, {\tt aQ} \tostate {\tt bQ} \, .
\]
We demonstrate that
\begin{enumerate}
\item[(1)]
$
(\Q_i, v_1, v_2) \lred{}_{\tt M} (\Q_j, v_1', v_2')\ 
\text{implies}\\ 
{\tt I}_M(\Q_i, \zero,\sem{v_1,v_2}_{\Q_i}^{\tt I}) \lred{}^*
{\tt I}_M(\Q_j, \zero,\sem{v_1',v_2'}_{\Q_j}^{\tt I}) 
$

\item[(2)]
${\tt I}_M$ does not transit to unreachable Minsky's states.\\ 
That is, if ${\tt I}_M(\Q_i, \zero,\sem{v_1,v_2}_{\Q_i}^{\tt I}) \lred{}^* {\tt I}_M(\S, \Sigma, \Psi) \lred{} {\tt I}_M(\Q_j, \Sigma', \Psi')$
with $\S \notin \{ \Q_0, \cdots , \Q_n \}$ then $\Sigma' = \zero$ and $\Psi' = \sem{v_1', v_2'}_{\Q_j}^{\tt I}$ and
$(\Q_i, v_1,v_2) \lred{}_{\tt M}^* (\Q_j, v_1', v_2')$.

\end{enumerate}

The proof of (1) is a case analysis on the type of Minsky instructions. 
We omit the simulation of the $\Inc$ instruction, which is simple, and discuss in detail
the simulation of $\DecJump$.
There are two cases (we assume the register is $R_1$):
\begin{description}
\item[$R_1 > 0$:] 
therefore $M$ performs the transition 
$(\Q_i, v_1, v_2) \lred{}_{\tt M} (\Q_j, v_1-1, v_2)$. Correspondingly, in ${\tt I}_M$,
an invocation of the function {\tt fdecQ$_i$} occurs. Then the contract performs the transitions
\[
\begin{array}{rl}
{\tt I}_M(\Q_i, \zero,\sem{v_1,v_2}_{\Q_i}^{\tt I}) \lred{{\tt fdecQ_i}} &
{\tt I}_M(\Q_i, {\it E}_1 | {\it E}_2 | {\it E}_3 \tostate {\tt dec}_1,\sem{v_1,v_2}_{\Q_i}^{\tt I})
\\
\lred{} &
{\tt I}_M({\tt dec}_1, \zero ,\sem{v_1,v_2}_{\Q_i}^{\tt I} | {\it E}_1 | {\it E}_2 | {\it E}_3 )
\end{array}
\]
where ${\it E_1}\! =\! 0 \, \event \, {\tt ackdec}_1 \tostate {\tt aQ_i}$,
${\it E_2}\! =\! 0 \, \event \, {\tt bQ_i} \tostate \Q_j$, and 
${\it E_3}\! =\! 0 \, \event \, {\tt aQ_j} \tostate {\tt bQ_j}$.
Since $v_1 > 0$, there is at least one event 
$0 \, \event \,  \withat{\tt dec}_1 \tostate {\tt ackdec}_1
$ in $\sem{v_1,v_2}_{\Q_i}^{\tt I}$. Then
{\small
\[
\begin{array}{rl}
{\tt I}_M({\tt dec}_1, \zero ,\sem{v_1,v_2}_{\Q_i}^{\tt I} | {\it E}_1 | {\it E}_2 | {\it E}_3 )
\lred{} &
{\tt I}_M({\tt dec}_1, \zero \tostate {\tt ackdec}_1,\sem{v_1\!\!-\!\!1,v_2}_{\Q_i}^{\tt I} | {\it E}_1 | {\it E}_2 | {\it E}_3) 
\\
\lred{} &
{\tt I}_M({\tt ackdec}_1, \zero, \sem{v_1\!\!-\!\!1,v_2}_{\Q_i}^{\tt I} | {\it E}_1 | {\it E}_2 | {\it E}_3 )
\\
\lred{} &
{\tt I}_M({\tt ackdec}_1, \zero \tostate {\tt aQ_i}, \sem{v_1\!\!-\!\!1,v_2}_{\Q_i}^{\tt I} | {\it E}_2 | {\it E}_3 )
\\
\lred{} &
{\tt I}_M({\tt aQ_i}, \zero , \sem{v_1\!\!-\!\!1,v_2}_{\Q_i}^{\tt I} | {\it E}_2 | {\it E}_3 )
\\
\lred{} &
{\tt I}_M({\tt aQ_i}, \zero \tostate {\tt bQ_i} , \sem{v_1\!\!-\!\!1,v_2}_{\Q_j}^{\tt I} | {\it E}_2 )
\\
\lred{} &
{\tt I}_M({\tt bQ_i}, \zero , \sem{v_1\!\!-\!\!1,v_2}_{\Q_j}^{\tt I} | {\it E}_2 )
\\
\lred{} &
{\tt I}_M({\tt bQ_i}, \zero \tostate \Q_j , \sem{v_1\!\!-\!\!1,v_2}_{\Q_j}^{\tt I} )
\\
\lred{} &
{\tt I}_M(\Q_j, \zero, \sem{v_1\!\!-\!\!1,v_2}_{\Q_j}^{\tt I} )
\end{array}
\]
}
\item[$R_1 = 0$:]
therefore $M$ performs the transition 
$(\Q_i, 0, v_2) \lred{}_{\tt M} (\Q_j, 0, v_2)$. Correspondingly, in ${\tt I}_M$,
an invocation of the function {\tt fzeroQ$_i$} occurs. Then the contract performs the transitions
\[
\begin{array}{rl}
{\tt I}_M(\Q_i, \zero,\sem{0,v_2}_{\Q_i}^{\tt I}) \lred{{\tt fzeroQ_i}} &
{\tt I}_M(\Q_i, {\it E}_1' | {\it E}_2 | {\it E}_3 \tostate {\tt dec}_1,\sem{v_1,v_2}_{\Q_i}^{\tt I} )
\\
\lred{} &
{\tt I}_M({\tt dec}_1, \zero ,\sem{0,v_2}_{\Q_i}^{\tt I} | {\it E}_1' | {\it E}_2 | {\it E}_3 )
\end{array}
\]
where ${\it E}_1' = 0 \, \event \, {\tt zero}_1 \tostate {\tt aQ_i}$ (${\it E_2}$ and 
${\it E_3}$ are as above).
Since $v_1 = 0$, there is no event 
$0 \, \event \,  \withat{\tt dec}_1 \tostate {\tt ackdec}_1$ in $\sem{v_1,v_2}_{\Q_i}^{\tt I}$. Then
it is possible to invoke {\tt fdec1}: 
\[
\begin{array}{rl}
{\tt I}_M({\tt dec}_1, \zero ,\sem{0,v_2}_{\Q_i}^{\tt I} | {\it E}_1' | {\it E}_2 | {\it E}_3 )
\lred{{\tt fdec1}} &
{\tt I}_M({\tt dec}_1, \zero \tostate {\tt zero}_1,\sem{0,v_2}_{\Q_i}^{\tt I} | {\it E}_1' | {\it E}_2 | {\it E}_3) 
\\
\lred{} &
{\tt I}_M({\tt zero}_1, \zero ,\sem{0,v_2}_\Q^{\tt I} | {\it E}_1' | {\it E}_2 | {\it E}_3) 
\\
\lred{} &
{\tt I}_M({\tt zero}_1, \zero \tostate {\tt aQ_i} ,\sem{0,v_2}_{\Q_i}^{\tt I} | {\it E}_2 | {\it E}_3)
\\
\lred{} &
{\tt I}_M({\tt aQ_i}, \zero ,\sem{0,v_2}_{\Q_i}^{\tt I} | {\it E}_2 | {\it E}_3)
\\
\lred{}^* & \text{\tt // similarly to the above case}
\\
\lred{} & {\tt I}_M(\Q_j, \zero, \sem{0,v_2}_{\Q_j}^{\tt I} )
\end{array}
\]
\end{description}


As regards (2), we assume that 
${\tt I}_M(\Q_i, \zero,\sem{v_1,v_2}_{\Q}^{\tt I}) \lred{}^* {\tt I}_M(\S, \Sigma, \Psi) \lred{} {\tt I}_M(\Q_j, \Sigma', \Psi')$ is such that  ${\tt I}_M(\Q_i, \zero,\sem{v_1,v_2}_{\Q}^{\tt I}) \lred{}^* {\tt I}_M(\S, \Sigma, \Psi)$ does not
contain pairs of transitions ${\tt I}_M(\S'', \Sigma'', \Psi'') \lred{} {\tt I}_M(\Q'', \Sigma''', \Psi''')$
such that $\S'' \notin \{\Q_0, \cdots , \Q_n\}$ and $\Q'' \in \{\Q_0, \cdots , \Q_n\}$. The general 
case follows by recurrence.

The proof is a case analysis on the first transition of ${\tt I}_M(\Q_i, \zero, \sem{v_1,v_2}_{\Q_i}^{\tt I})$:
\begin{itemize}

\item
if ${\tt I}_M(\Q_i, \zero,\sem{v_1,v_2}_\Q^{\tt I}) \lred{} {\tt I}_M(\Q_i, \zero,\zero)$ is an instance of \rulename{Tick}.
This case is vacuous: the continuations of ${\tt I}_M(\Q_i, \zero,\zero)$ may be either instances of 
\rulename{Tick} or the invocation of exactly one function ${\tt f} \in \{ {\tt fincQ_i}, {\tt fdecQ_i}, {\tt fzeroQ_i}
\}$ possibly followed by an invocation of either ${\tt fdec}_1$ or ${\tt fdec}_2$ and by the execution of 
the event $0 \event {\tt zero}_1 \tostate {\tt aQ_i}$ or $0 \event {\tt zero}_2 \tostate {\tt aQ_i}$. Therefore 
no transition ${\tt I}_M(\S, \Sigma, \Psi) \lred{} {\tt I}_M(\Q_j, \Sigma', \Psi')$ will be ever possible.

\item
if ${\tt I}_M(\Q_i, \zero,\sem{v_1,v_2}_{\Q_i}^{\tt I}) \lred{{\tt fincQ_i}} {\tt I}_M(\Q_i, {\it Ev} \tostate 
{\tt aQ_i},\sem{v_1,v_2}_{\Q_i}^{\tt I} )$ then, according to Table~\ref{tab.encodingI}, $M$ contains the instruction 
$\Q_i : \Inc(R_1, \Q_j)$ or $\Q_i : \Inc(R_2, \Q_j)$.
We consider the first case, i.e., $R_1$ is incremented
(the argument for $R_2$ is similar). 
It is easy to verify that
${\tt I}_M(\Q_i, {\it Ev} \tostate {\tt aQ_i},\sem{v_1,v_2}_{\Q_i}^{\tt I} ) \lred{}^5
{\tt I}_M(\Q_j, \zero,\sem{v_1+1,v_2}_{\Q_j}^{\tt I})$
with a deterministic computation that does not traverse configurations 
${\tt I}_M(S'', \Sigma'', \Psi'')$ with $S'' \in \{\Q_0, \cdots , \Q_n\}$.

\item
if ${\tt I}_M(\Q_i, \zero,\sem{v_1,v_2}_{\Q_i}^{\tt I}) \lred{{\tt fdecQ_i}} {\tt I}_M(\Q_i, {\it Ev} \tostate 
{\tt aQ_i},\sem{v_1,v_2}_{\Q_i}^{\tt I} )$ then, according to Table~\ref{tab.encodingI}, $M$ contains 
$\Q_i : \DecJump(R_1, \Q', \Q'')$ or $\Q_i : \DecJump(R_2, \Q', \Q'')$.
We consider the first case, i.e., $R_1$ is considered
(the argument for $R_2$ is similar).
There are two subcases:
\begin{description}
\item[$(v_1 > 0) \; $]  
In this case it is easy to verify that
${\tt I}_M(\Q_i, {\it Ev} \tostate {\tt aQ_i},\sem{v_1,v_2}_{\Q_i}^{\tt I} ) \lred{}^8
{\tt I}_M(\Q'', \zero, \sem{v_1-1,v_2}_{\Q''}^{\tt I})$
with a computation that does not traverse configurations 
${\tt I}_M(S'', \Sigma'', \Psi'')$ with $S'' \in \{\Q_0, \cdots , \Q_n\}$.
In this case, the computation is not deterministic because the
third transition is the invocation of ${\tt fdec}_1$. 
Actually, it is also possible
to perform a transition \rulename{Tick}. Analogous to the first case
above, after the transition \rulename{Tick}, 
no transition ${\tt I}_M(\S, \Sigma, \Psi) \lred{} {\tt I}_M(\Q_j, \Sigma', \Psi')$ will be ever possible.
%
%
%
%

\item[$(v_1 = 0 ) \; $]
This case is vacuous because we have
${\tt I}_M(\Q_i, {\it Ev} \tostate 
{\tt dec}_1,\sem{0,v_2}_{\Q_i}^{\tt I} )$ $\lred{}^2$ ${\tt I}_M({\tt dec}_1,  \zero ,\sem{0,v_2}_{\Q_i}^{\tt I}  ~|~ {\it Ev})$ 
and, for every 
${\tt I}_M({\tt dec}_1 , \zero ,\sem{0,v_2}_{\Q_i}^{\tt I}  ~|~ {\it Ev})$ $\lred{}^* {\tt I}_M(\S, \Sigma,\Psi)$, 
it is easy to verify that $\S \notin 
\{\Q_0, \cdots , \Q_n\}$. 
\end{description}

\item
if ${\tt I}_M(\Q_i, \zero,\sem{v_1,v_2}_{\Q_i}^{\tt I}) \lred{{\tt fzeroQ_i}} {\tt I}_M(\Q_i, {\it Ev} \tostate 
{\tt dec}_1,\sem{v_1,v_2}_{\Q_i}^{\tt I} )$  then, according to Table~\ref{tab.encodingI}, $M$ contains 
$\Q_i : \DecJump(R_1, \Q', \Q'')$ or $\Q_i : \DecJump(R_2, \Q', \Q'')$.
We consider the first case, i.e., $R_1$ is considered
(the argument for $R_2$ is similar).
There are two subcases:
\begin{description}
\item[$(v_1 = 0) \; $]  
We have that 
${\tt I}_M(\Q_i, {\it Ev} \tostate 
{\tt dec}_1,\sem{v_1,v_2}_{\Q_i}^{\tt I} )
 \lred{}^9$ ${\tt I}_M(\Q', \zero,\sem{0,v_2}_{\Q'}^{\tt I})$. As in the case of {\tt fdecQ$_i$}, the computation is not deterministic due to the
presence of alternative \rulename{Tick} transitions, but after a \rulename{Tick} 
no transition ${\tt I}_M(\S, \Sigma, \Psi) \lred{} {\tt I}_M(\Q_j, \Sigma', \Psi')$ will be ever possible.


\item[$(v_1 > 0) \; $]
Vacuous, analogous to the case of {\tt fdecQ$_i$} with $v_1 = 0$.
\end{description}
\end{itemize}

\Reviewer{1}{model $I_M$ could contain the last paragraph from the start so the connection is clearer; it is more like an afterthought}
\Answer{We have not understood the comment. May the Reviewer be clearer?}
To conclude, given a Minsky machine $M$ with a final state $\Q_F$, 
let ${\tt I}_M$ be the corresponding
{\miniStipulaI} contract according to Table~\ref{tab.encodingI}. Let also ${\tt I}_M'$  be the
contract ${\tt I}_M$ extended with a clause $\clause{\Q_F}{{\tt f}}{\Q_F'}$, where $\Q_F'$ is 
a new state.
Since the reachability of $(\Q_F, v_1, v_2)$, for every $v_1$, $v_2$, is undecidable for Minsky machines,
because of the foregoing properties 
we derive that the reachability of $\clause{\Q_F}{{\tt f}}{\Q_F'}$ is also undecidable in 
{\miniStipulaI} (similarly for events).
\qed
\end{proof}

\thmtimeahead*

\begin{proof}
Let $M$ be the Minsky machine with states $\{\Q_1, \cdots, \Q_n\}$ 
and ${\tt TA}_M$ be the encoding in {\miniStipulaTA}
as defined in Table~\ref{tab.encodingTA}. 
Let 
\[
\sem{v_1,v_2}_\Q^{\tt TA} = \underbrace{1 \, \event \, \mathtt{dec}_1 \tostate {\tt ackdec}_1}_{v_1 \text{ times}}
\; | \; \underbrace{1 \, \event \,  \mathtt{dec}_2 \tostate {\tt ackdec}_2}_{v_2 \text{ times}}
\; | \; 1 \event \Q \tostate {\tt end} \; .
\]
and let $\Psi_{\tt TA}$, $\Psi_{\tt TA}'$ be either $\zero$ or ($i$ is either 1 or 2)
\[
\begin{array}{l}
0 \event {\tt dec}_1 \tostate {\tt end} \, | \, 0 \event {\tt dec}_2 \tostate {\tt end}
\, | \, 0 \event {\tt nextQb}_i \tostate {\tt end}
\, | \, 0 \event {\tt ackdec}_1 \tostate \withat{\tt end}
\\
| \, 0 \event \, {\tt ackdec}_2 \tostate {\tt end} \; .
\end{array}
\]
(the cases where $\Psi$ and $\Psi'$ are not $\zero$ correspond to those where the machine transition 
is a $\DecJump$).
We need to prove
\begin{enumerate}
\item[(1)]
$(\Q_i, v_1, v_2) \lred{}_{\tt M} (\Q_j, v_1', v_2')$
implies\\
${\tt TA}_M(\Q_i, \zero,\sem{v_1,v_2}_{\Q_i}^{\tt TA} \, | \, \Psi) \lred{}^*
{\tt TA}_M(\Q_j, \zero,\sem{v_1',v_2'}_{\Q_j}^{\tt TA} \, | \, \Psi')$ 


\item[(2)]
${\tt TA}_M$ does not transit to unreachable Minsky's states.\\ 
That is, if ${\tt TA}_M(\Q_i, \zero,\sem{v_1,v_2}_{\Q_i} ~|~ \Psi_{\tt TA}) \lred{}^* {\tt TA}_M(\S, \Sigma, \Psi) 
\lred{} {\tt TA}_M(\Q_j, \Sigma', \Psi')$
with $\S \notin \{ \Q_0, \cdots , \Q_n \}$ then $\Sigma' = \zero$ and $\Psi' = \sem{v_1', v_2'}_{\Q_j}^{\tt TA}~|~ 
\Psi_{\tt TA}'$ and also
$(\Q_i, v_1,v_2) \lred{}_{\tt M}^* (\Q_j, v_1', v_2')$. 
\end{enumerate}

The proofs of (1) and (2) consist of a detailed case analysis on the transitions of ${\tt TA}_M$.
The proof of (1) is similar to the corresponding one of Theorem~\ref{thm.instantaneous}, therefore omitted. As
regards (2), the difference with that of Theorem~\ref{thm.instantaneous} is the fact 
that ${\tt TA}_M(\Q_i, \zero,\sem{v_1,v_2}_{\Q_i}^{\tt TA} ~|~ \Psi_{\tt TA}) \lred{}^* {\tt TA}_M(\S, \Sigma, \Psi) 
\lred{} {\tt TA}_M(\Q_j, \Sigma', \Psi')$ contains exactly 3 transitions \rulename{Tick} (in the basic case
where no state $\{ \Q_0, \cdots , \Q_n \}$
is traversed). Otherwise one has to demonstrate that the case becomes vacuous.
The details are omitted. In the same way we may conclude the proof.
\qed
\end{proof}

\thmdeterminate*

\begin{proof}
Let $M$ be the Minsky machine and ${\tt D}_M$ be the encoding in {\miniStipulaD}
as defined in Tables~\ref{tab.encodingD} and~\ref{tab.encodingD_addendum}. 
After a sequence of transitions \rulename{Tick}, the contract ${\tt D}_M$ performs the
transitions
\[
{\tt D}_M({\tt Start}, \zero,\zero) \lred{{\tt fstart}} 
{\tt D}_M({\tt Start}, {\it Ev} \tostate \Q_0 ,\zero) \lred{} {\tt D}_M(\Q_0, \zero,{\it Ev})
\]
where $\Q_0$ is the initial state of $M$ and ${\it Ev} = 0 \event {\tt notickA} \tostate {\tt cont}$. Next, let 
\[
\sem{v_1,v_2}_\X^{\tt D} = \underbrace{1 \, \event \, \mathtt{dec}_1 \tostate {\tt ackdec}_1}_{v_1 \text{ times}}
\; | \; \underbrace{3 \, \event \,  \mathtt{dec}_2 \tostate {\tt ackdec}_2}_{v_2 \text{ times}}
\; | \; 0 \event {\tt notickX} \tostate {\tt cont} 
\]
where ${\tt X}$ may be either $\A$ or $\B$.
We need to prove
\begin{enumerate}
\item[(1)]
$
(\Q_i, v_1, v_2) \lred{}_{\tt M} (\Q_j, v_1', v_2') \ \
\text{implies} \\
\quad {\tt D}_M(\Q_i, \zero,\sem{v_1,v_2}_{\tt X}^{\tt D} ) \lred{}^*
{\tt D}_M(\Q_j, \zero,\sem{v_1',v_2'}_{\tt Y}^{\tt D} ) 
$
where {\tt X} and {\tt Y} are either {\tt A} or {\tt B};

\item[(2)]

${\tt D}_M$ does not transit to unreachable Minsky's states. \\
That is, if ${\tt D}_M(\Q_i, \zero,\sem{v_1,v_2}_{\X}^{\tt D} ) \lred{}^* {\tt D}_M(\S, \Sigma, \Psi) 
\lred{} {\tt D}_M(\Q_j, \Sigma', \Psi')$,
with $\S \notin \{ \Q_0, \cdots , \Q_n \}$, then $\Sigma' = \zero$ and $\Psi' = 
\sem{v_1', v_2'}_{{\tt Y}}^{\tt D}$ and
$(\Q_i, v_1,v_2) \lred{}_{\tt M}^* (\Q_j, v_1', v_2')$. 

\end{enumerate}

%
%
Like Theorems~\ref{thm.instantaneous} and~\ref{thm.timeahead}, the proofs of (1) and (2) are 
a case analysis on the possible transitions of ${\tt D}_M$. They are 
omitted because  
similar to the foregoing ones. We may therefore conclude the proof.
\qed
\end{proof}

\subsection{Proofs of Section \ref{sec:det_instantaneous}}

\propcorrespondence*

\begin{proof}
The proofs of statements (\emph{i}) and (\emph{ii}) are case analyses on the possible transition 
rules that can be applied. They are straightforward and omitted. 
As regards (\emph{iii}),
we notice that no rule \rulename{Function} can be applied because $\Q \in \InitEv(\C)$, hence $\Q$
cannot be  an initial state of function in {\miniStipulaDI} and {\miniStipulaDIP}. \rulename{State-Change}
cannot be applied because $\Sigma = \zero$ and \rulename{Event-Match} cannot be applied because 
$\noredbis{\SemEvent, \Q}$. Therefore, the unique rule that can be applied is \rulename{Tick} and we have
\[
\pairbis{\C(\Q \semi \zero \semi \Psi)}{\time} \lred{} \pairbis{\C(\Q \semi \zero \semi \zero)}{\time}
\lred{}^* \pairbis{\C(\Q \semi \zero \semi \zero)}{\time}
\]
where $\pairbis{\C(\Q \semi \zero \semi \zero)}{\time}
\lred{}^* \pairbis{\C(\Q \semi \zero \semi \zero)}{\time}$ are instances of \rulename{Tick}, therefore 
$\C(\Q \semi \zero \semi \Psi)$ is stuck.
Correspondingly, since \rulename{Tick} is replaced by \rulename{Tick-Plus} in {\miniStipulaDIP},
no transition $\lredTickP{}$ is possible. Hence $\pairbis{\C(\Q \semi \Sigma \semi \Psi)}{\time} \nolredTickP{}$.
\qed
\end{proof}

\lemwellstructuredTP*

\begin{proof}
To verify that $(\mathcal{C}, \lredTickP{}, \preceq)$ 
is well-structured we need to demonstrate the properties (1) and (2) of Definition~\ref{def.wsts}:

\medskip

\noindent
\emph{(1) $\preceq$ is a well-quasi ordering}.
$\Psi$ are multisets of events and the relation $\preceq$ is multiset inclusion.
Furthermore, by definitions the set of all possible events is always finite.
Therefore, by Dickson's Lemma \cite{Dickson}, $\preceq$ is a well-quasi ordering.

\medskip

\noindent
\emph{(2) $\preceq$ is upward compatible with $\lredTickP{}$}.
We demonstrate that, if $\C(\Q,\Sigma,\Psi) \lredTickP{} \C(\Q',\Sigma',\Psi')$ and $\C(\Q,\Sigma,\Psi) \preceq
\C(\Q,\Sigma,\Psi'')$ then there is $\Psi'''$ s.t. $\C(\Q,\Sigma,\Psi'') \lredTickP{} 
\C(\Q',\Sigma',\Psi''')$ with $\C(\Q,\Sigma,\Psi') \preceq \C(\Q,\Sigma,\Psi''')$. The proof consists of
a case analysis on the rule used in the transition $\C(\Q,\Sigma,\Psi) \lredTickP{} \C(\Q',\Sigma',\Psi')$.
When the rule is an instance of \rulename{Function} or \rulename{State-Change} the upward compatibility is
immediate. If the rule is an instance of \rulename{Tick-Plus} then $\Q \notin \InitEv(\C)$. Hence the
same rule may be applied to $\C(\Q,\Sigma,\Psi'')$; in this case $\Psi' = \zero = \Psi'''$
(because the events in $\Psi$ and $\Psi''$ have time expressions 0). When the rule is an instance of 
\rulename{Event-Match} then let $0 \event_{\tt n} \Q \tostate \Q''$ be the event in $\Psi$ that has 
been scheduled. Since $\Psi \preceq \Psi''$, the same event can be scheduled in $\C(\Q,\Sigma,\Psi'')$. 
By definition, we have $\Psi' \preceq \Psi'''$.
\qed
\end{proof}
\lemfinitepredbasis*
\begin{proof}
Regarding 1, the algorithm for $\preceq$ is trivial, as well as its decidability. The algorithm for 
$\lredTickP{}$ is defined by the operational semantics, henceforth its decidability.

\Reviewer{1}{the reference Higman's Lemma is confusing here (and the citation is wrong): usually, it refers to the fact that every infinite sequence of words has an infinite ascending subsequence.}
\Answer{Repaired}

Regarding 2, 
we know that since $\preceq$ is a well-quasi-ordering
the set $\uparrow \! \contract$ has a finite basis. Let $\mathcal{B}$ be such basis, 
we prove  that $\Pred(\mathcal{B})$ has a finite basis that is also computable.
Let $\contract' \in \mathcal{B}$; there are three cases:
\begin{description}
\item[$\contract' = \C(\Q', \Sigma, \Psi) \text{ and } \Sigma \neq \zero$:]
In this case, 
the transition $\contract'' \lredTickP{} \contract'$ con be obtained in two
possibile ways: 
by applying \rulename{function} 
or 
\rulename{Event-Match}. In the first case,
$\Pred(\mathcal{B})$ contains all the tuples $\C(\Q', \zero, \Psi)$
for which there is a
function $\Qwithat' \;\f\, \{\, W\,\} \,\tostate\, \Qwithat''  \in \C$
that can produce $\Sigma$; in the second case,
$\Pred(\mathcal{B})$ contains all the tuples $\C(\Q', \zero, \Psi')$
where $\Psi'$ extends $\Psi$ with an 
event $\mathtt{0} \event 
\Qwithat' \tostate \Qwithat'' \in \C$ that can produce $\Sigma$.
In both cases, $\Pred(\mathcal{B})$ is finite and computable. 

\item[$\contract' = \C(\Q', \zero, \Psi) \text{ and } \Psi \neq \zero$:]
The unique rule that gives a transition $\contract'' \lredTickP{} \contract'$ is \rulename{State-Change}. 
In this case we have
$\contract'' = \C(\Q, W \tostate \Q', \Psi')$ where 
$W = \bigl( \mathtt{0} \event_{\! \n_i} \Q_i \,\tostate \Q_i' \bigr)^{i \in 1..h}$
and $\Psi = \Psi' | \mathtt{0} \event_{\! \n_1} \Q_1 \,\tostate \Q_1' | \cdots | 
\mathtt{0} \event_{\! \n_h} \Q_h \,\tostate \Q_h'$. There are two possible cases for 
obtaining the terms $W \tostate \Q'$:
(\emph{i}) the functions $\Qwithat \;\f\, \{\, W\,\} \,\tostate\, \Qwithat' \in \C$,
for every possible subterm $W$ of $\Psi$; and  (\emph{ii}) the events $\mathtt{0} \event 
\Qwithat \tostate \Qwithat' \in \C$ (in this sub-case, $W = \zero$). In both cases,
$\Pred(\mathcal{B})$ is finite and computable.

\item[$\contract' = \C(\Q', \zero, \zero)$:]
Then the transition $\contract'' \lredTickP{} \contract'$ is an instance of \rulename{Tick-Plus}.
$\contract''$ could be any tuple of type $\C(\Q', \zero, \Psi')$.
$\{ \C(\Q', \zero, \zero) \}$ is a finite basis for all such tuples.
\qed
\end{description}
\end{proof}

\end{document}